\documentclass[11pt]{article}

\usepackage{amsmath, amssymb, amsthm, amsfonts,bm,mathtools}
\usepackage[inline]{enumitem}
\usepackage{csquotes}
\usepackage{hhline}
\usepackage{cite}
\usepackage{framed}
\usepackage[framemethod=tikz]{mdframed}
\usepackage{graphicx}
\usepackage{color}
\usepackage{multirow}
\usepackage{algorithm, subcaption}
\usepackage[noend]{algpseudocode}
\usepackage{amssymb}
\usepackage{float}
\setitemize{noitemsep,topsep=3pt,parsep=3pt,partopsep=3pt}

\usepackage[margin = 1in]{geometry}

\usepackage{xspace}

\definecolor{darkgreen}{rgb}{0,0.5,0}
\definecolor{darkblue}{rgb}{0,0,0.8}
\usepackage{hyperref}
\hypersetup{
    unicode=false,          
    colorlinks=true,        
    linkcolor=darkblue,          
    citecolor=darkgreen,        
    filecolor=magenta,      
    urlcolor= black           
}
\RequirePackage[]{silence}
\usepackage[nameinlink,capitalize]{cleveref}

\theoremstyle{theorem}
\newtheorem{theorem}{Theorem}[section]
\theoremstyle{lemma}
\newtheorem{lemma}[theorem]{Lemma}
\theoremstyle{corollary}
\newtheorem{corollary}[theorem]{Corollary}
\theoremstyle{claim}

\theoremstyle{proposition}
\newtheorem{proposition}[theorem]{Proposition}

\theoremstyle{observation}
\newtheorem{observation}[theorem]{Observation}

\theoremstyle{definition}

\theoremstyle{remark}

\newcommand{\einit}{e_{\text{init}}}

\newcommand{\ignore}[1]{}

\algnewcommand\algorithmicswitch{\textbf{switch}}
\algnewcommand\algorithmiccase{\textbf{case}}

\algdef{SE}[SWITCH]{Switch}{EndSwitch}[1]{\algorithmicswitch\ #1\ \algorithmicdo}{\algorithmicend\ \algorithmicswitch}%
\algdef{SE}[CASE]{Case}{EndCase}[1]{\algorithmiccase\ #1}{\algorithmicend\ \algorithmiccase}%
\algtext*{EndSwitch}%
\algtext*{EndCase}%
\newcommand{\eps}{\varepsilon}
\newcommand{\arb}{a}
\newcommand{\cs}{\mathfrak C}

\newcommand{\congest}{\ensuremath{\mathsf{CONGEST}~}}
\newcommand{\local}{$\mathsf{LOCAL}$\xspace}
\newcommand{\slocal}{$\mathsf{SLOCAL}$\xspace}

\newcommand{\poly}{\operatorname{\text{{\rm poly}}}}

\newcommand{\polylog}{\operatorname{\text{{\rm polylog}}}}
\newcommand{\polyloglog}{\operatorname{\text{{\rm polyloglog}}}}

\newcommand{\Prob}[1]{\Pr\left(#1\right)}

\DeclareMathOperator{\E}{\mathbb{E}}

\def\cut{\mbox{\tt CUT}}

\renewcommand{\paragraph}[1]{\vspace{0.15cm}\noindent {\bf #1}:}

\begin{document}

\title{On the Locality of Nash-Williams Forest Decomposition and Star-Forest Decomposition\footnote{This is an extended version of a paper appearing in the ACM Symposium on Principles of Distributed Computing (PODC) 2021}}
\author{David G. Harris \\ University of Maryland \\ \url{davidgharris29@gmail.com} \and Hsin-Hao Su \footnote{supported by NSF Grant No. CCF-2008422.}\\ Boston College \\ \url{suhx@bc.edu} \and  Hoa T.~Vu \\ San Diego State University \\ \url{hvu2@sdsu.edu}}	
\date{}

\maketitle


\begin{abstract}
Given a graph $G=(V,E)$ with arboricity $\arb$, we study the problem of decomposing the edges of  $G$ into  $(1+\eps)\arb$ disjoint forests  in the distributed \local model. Here $G$ may be a simple graph or multi-graph. While there is a polynomial time centralized algorithm for $\arb$-forest decomposition  (e.g. [Imai, J. Operation Research Soc. of Japan `83]), it remains an open question how close we can get to this exact decomposition  in the \local model.

Barenboim and Elkin [PODC `08] developed a \local algorithm to compute a $(2+\eps)\arb$-forest decomposition in $O(\frac{\log n}{\eps})$ rounds. Ghaffari and Su [SODA `17] made further progress by computing a $(1+\eps) \arb$-forest decomposition in $O(\frac{\log^3 n}{\eps^4})$ rounds when $\eps \arb = \Omega(\sqrt{\arb \log n})$, i.e., the limit of their algorithm is an $(\arb+ \Omega(\sqrt{\arb \log n}))$-forest decomposition. This algorithm, based on a combinatorial construction of Alon, McDiarmid \& Reed [Combinatorica `92], in fact provides a decomposition of the graph into \emph{star-forests}, i.e., each forest is a collection of stars.

Our main goal is to reduce the threshold of $\eps \arb$  in $(1+\eps)\arb$-forest decomposition.  We obtain a number of results with different parameters; as some notable examples, we get:
\begin{itemize}
\item An $O(\frac{\Delta^{\rho} \log^4 n}{ \eps})$-round algorithm when  $\eps\arb = \Omega_{\rho}(1)$ in multigraphs, where $\rho > 0$ is any arbitrary constant.
\item An $O(\frac{\log^4 n \log \Delta}{\eps})$-round algorithm when $\eps\arb = \Omega(\frac{\log \Delta}{\log \log \Delta})$  in multigraphs.
\item An $O(\frac{\log^ 4 n}{\eps})$-round algorithm when $\eps\arb = \Omega(\log n)$ in multigraphs. This also covers an extension of the forest-decomposition problem to list-edge-coloring.
\item An $O(\frac{\log^3 n}{\eps})$-round algorithm for star-forest decomposition for $\eps\arb = \Omega(\sqrt{\log \Delta} + \log \arb)$ in simple graphs. When $\eps \arb \geq \Omega(\log \Delta)$, this also covers a list-coloring variant.
\end{itemize} 

Our techniques also give an algorithm for $(1+\eps) \arb$-outdegree-orientation in $O(\frac{\log^3 n}{\eps})$ rounds, which is the first algorithm with {\it linear dependency} on $\eps^{-1}$. 

At a high level, the first three results come from a combination of network decomposition, load balancing, and a new structural result on local augmenting sequences. The fourth result uses a more careful probabilistic analysis for the construction of Alon, McDiarmid \& Reed; the bounds on star-forest-decomposition were not previously known even non-constructively.
\end{abstract}
\thispagestyle{empty}
\newpage
\pagenumbering{arabic} 
\section{Introduction}
\sloppy
Consider a loopless (multi-)graph $G = (V,E)$ with $n = |V|$ vertices, $m = |E|$ edges, and maximum degree $\Delta$. A $k$-forest decomposition (abbreviated $k$-FD) is a partition of the edges into $k$ forests. The \emph{arboricity} of $G$, denoted $\arb(G)$, is a measure of sparsity defined as the minimum number $k$ for which a $k$-forest decomposition of $G$ exists. We also write $\arb(E)$ or just $\arb$ when $G$ is understood.   An elegant result of Nash-Williams \cite{NW64} shows that $\arb(G)$ is given by the formula:
$$
\arb(G) = \max_{\substack{H \subseteq G \\ |V(H)| \geq 2 }} \left\lceil \frac{ |E(H)| }{|V(H)| - 1} \right\rceil.
$$
Note that the RHS is clearly a lower bound on $\arb$ since each forest can consume at most $|V(H)| - 1$ edges in a subgraph $H$.

Forest decomposition can be viewed as a variant of  proper edge coloring: in the latter problem, the edges should be partitioned into matchings, while in the former they should be partitioned into forests. Like edge coloring, forest decomposition has applications to scheduling radio or wireless networks \cite{RL93,GKJ07}.  In the centralized setting, a series of polynomial-time algorithms have been developed to compute $\arb$-forest decompositions \cite{HI83,RT85,GS85,GW92}.  

In this work, we study the problem of computing forest decompositions in the \local model of distributed computing \cite{linial92}. In this model, the vertices operate in synchronized rounds where each vertex sends and receives messages of arbitrary size to its neighbors, and performs arbitrary local computations. Each vertex also has a unique ID which is a binary string of length
$O(\log n)$. An $r$-round \local algorithm implies that each vertex only uses information in its $r$-hop neighborhood to compute the answer, and vice versa.

There has been growing interest in investigating the gap between  efficient computation in the \local model and the existential bounds of various combinatorial structures. For example, consider proper edge coloring. Vizing's classical result \cite{vizing1964estimate} shows that there exists a $(\Delta + 1)$-edge-coloring in simple graphs. A long series of works have developed  \local algorithms using smaller number of colors \cite{PS97,dubhashi1998near,EPS15,ChangHLPU18,GKMU18,SV19}. This culminated with a $\poly(\Delta, \log n)$-round algorithm in \cite{bernshteyn2020fast} for $(\Delta+1)$-edge-coloring, matching the existential bound.

Computing an $\arb$-forest decomposition in the \local model requires $\Omega(n)$ rounds even in simple graphs with constant $\Delta$ (see Proposition~\ref{cor:lower_bound_simple}). Accordingly, we aim for    $(1+\eps) \arb$ forests, i.e. $\eps \arb$ excess forests beyond the $\arb$ forests 
required existentially. Beside round complexity, a key objective is to minimize the value $\eps \arb$.

The first results in the \local model were due to  \cite{BE10}, who developed an $O(\frac{\log n}{\eps})$-round algorithm for $(2+\eps)\arb$-FD along with a lower bound of $\Omega( \frac{\log n}{\log \arb} - \log^* n)$ rounds for $O(a)$-FD.  These have been building block in many distributed and parallel algorithms \cite{BE10, BE11, BBDFHKU19, Kuhn20, SDS20}.      Open Problem 11.10 of \cite{BE03} raised the question of whether it is possible to use fewer than $2 \arb$ forests.  Ghaffari and Su \cite{GS17} made some progress with a randomized algorithm for $(1+\eps)\arb$-FD in $O(\log^3 n/ \eps^4)$ rounds in simple graphs when $\eps = \Omega(\sqrt{\log n / \arb})$, i.e., the minimum number of obtainable forests is $a + \Omega(\sqrt{a \log n})$.  

We make further progress with a randomized algorithm  for $(\arb + 3)$-FD in $\poly(\Delta, \log n)$ rounds in multigraphs.  The polynomial dependence on $\Delta$ can be removed when $\eps \arb$ is larger; for example, we obtain a $(1+\eps)\arb$-FD in $O(1/\eps) \cdot \polylog(n)$ rounds for $\eps \arb = \Omega(\log \Delta / \log \log \Delta)$.

\paragraph{List Forest Decomposition}
Similar to edge coloring, there is a list version of the forest decomposition problem: each edge $e$ has a color-palette $Q(e)$ and should choose a color $\phi(e) \in Q(e)$ so that, for any color $c$, the subgraph induced by the $c$-colored edges forms a forest. We refer to this as \emph{list-forest decomposition} (abbreviated \emph{LFD}). We denote by $\cs = \bigcup_e Q(e)$ the set of all possible colors; this generalizes $k$-forest-decomposition, which can be viewed as the case where $ \cs = \{1, \dots, k \}$.

Based on general matroid arguments, Seymour \cite{seymour1998note} showed that an LFD exists whenever the palettes all have size at least $\arb$. The total number of forests (one per color) may then be much larger than $\arb$; in this case, the excess is measured in terms of the number of extra colors in edges' palettes (in addition to the $\arb$ colors required by the lower bound).

Seymour's construction can be turned into a polynomial-time centralized algorithm with standard matroid techniques.  However, these do not extend to the  \local model. As a proof of concept, we give $\poly(\log n, 1/\eps)$-round algorithms when palettes have size $(1+\eps) \arb$ for $\eps \arb = \Omega(\min\{ \log n, \sqrt{\arb \log \Delta} \})$. A key open problem is to find an efficient algorithm for $\eps \arb \geq \Omega(1)$.

\paragraph{Low-Diameter and Star-Forest Decompositions} We say the decomposition has \emph{diameter $D$} if every tree in every $c$-colored forest has strong diameter at most $D$.  Minimizing $D$ is interesting from both practical and theoretical aspects. For example, given a $k$-FD of diameter $D$, we can find an orientation of the edges to make it into $k$ rooted forests in $O(D)$ rounds of the \local model.

 In the extreme case $D = 2$, each forest is a collection of stars, i.e., a \emph{star-forest}. This has received some attention in combinatorics.  We refer to this as \emph{$k$-star-forest decomposition} (abbreviated $k$-SFD); we give an $O(\frac{\log^3 n}{\eps})$-round algorithm for $(1+\eps)\arb$-SFD when $\eps \arb = \Omega(\log a + \sqrt{\log \Delta})$ in simple graphs.  The algorithm also solves the list-coloring variant, which we call  \emph{list-star-forest-decomposition} (abbreviated LSFD), when $\eps \arb = \Omega(\log \Delta)$.

For larger diameters, we show how to convert an arbitrary $k$-FD into a $(1+\eps)k$-FD with diameter $D = O(\log n/\eps)$; when $\eps k$ is large enough, the diameter can be reduced further to $D = O(1/\eps)$, which is optimal (see Proposition~\ref{eps-bound}).

\subsection{Summary of Results}
Our results for forest decomposition balance a number of measures: the number of excess colors required, the running time, the tree diameters, LFD versus FD, and multigraphs versus simple graphs. Table~\ref{fi1} below summarizes a number of parameter combinations.

 Here, $\rho > 0$ represents any desired constant and we use $\Omega_{\rho}$ to represent a constant term which may depend on $\rho$.  Thus, for instance, the final listed algorithm requires excess $\overline K \log \Delta$ and the third listed algorithm requires excess $\overline K_{\rho}$, where $\overline K$ and $\overline K_{\rho}$ are universal constants.

\begin{table}[H]
\centering
\begin{tabular}{|l|l|l|l|l|}
\hline
Excess colors & Lists?  & Multigraph? & Runtime & Forest Diameter  \\
\hline
\hline
3 & No & Yes & $O( \Delta^2 \arb \log^4 n \log \Delta)$ & $\leq n$ \\
\hline
$\geq 4$  & No & Yes & $O( \Delta^2 \log^4 n \log \Delta / \eps)$ & $O(  \log n / \eps)$ \\
  \hline
  $\Omega_{\rho}(1)$ & No & Yes & $O( \Delta^{\rho} \log^4 n / \eps )$ &  $O( \log n / \eps )$ \\ 
  \hline
    $\Omega_{\rho}(  \frac{\log \Delta}{\log \log \Delta} )$ & No & Yes & $O_{\rho}( \log^4 n \log^{\rho} \Delta / \eps )$ &  $O( \log n / \eps )$ \\ 
    \hline
    $\geq 4 + \rho \log \Delta$ & No & Yes & $O_{\rho}( \log^4 n / \eps )$ &  $O( \log n / \eps)$ \\ 
    \hline    
        $\Omega( \sqrt{\arb \log \Delta} )$ & No & Yes & $O( \log^4 n / \eps )$ &  $O( 1/ \eps )$ \\ 
        \hline 
  $\Omega(\log n)$ & No & Yes & $O( \log^3 n / \eps )$ & $O(  1/ \eps )$ \\
        \hline        
        $\Omega( \sqrt{\arb \log \Delta} )$ & Yes & Yes & $O( \log^4 n / \eps^2 )$ &  $O( \log n / \eps^2)$ \\ 
        \hline
  $\Omega(\log n)$ & Yes & Yes & $O( \log^4 n / \eps )$ & $O(  \log n/ \eps )$ \\
  \hline
  $\Omega(\sqrt{\log \Delta} + \log a)$ & No & No & $ O(\log^3 n / \eps)$ & 2 (star) \\
  \hline
   $\Omega( \log \Delta)$ & Yes & No & $ O(\log^3 n / \eps)$ & 2 (star) \\
   \hline
\end{tabular}
\vspace{1mm}
\caption{Possible algorithms for forest decompositions of $G$}
\label{fi1}
\vspace{-1mm}
\end{table}

We also show that $\Omega(1/\eps)$ rounds are needed for $(1+\eps) \arb$-FD in  multigraphs (see Theorem \ref{thm:lower_bound_multi}).

\paragraph{Note on deterministic algorithms} All the algorithms we consider (unless specifically stated otherwise) are randomized algorithms which succeed {\it with high probability} (abbreviated w.h.p.), i.e. with probability at least $1 - 1/\poly(n)$. It will turn out that the algorithms we develop have the property that if the algorithm fails (i.e. the output does not satisfy desired properties), then this can be detected by a node checking its local neighborhood during the algorithm run. Such randomized algorithms are referred to as \emph{Las Vegas algorithms} in \cite{GHK18}. Using a recent breakthrough of \cite{GHK18, RG20}, such Las Vegas algorithms can be automatically derandomized with an additional $\polylog(n)$ factor in the runtime. For brevity, we will not explicitly show that the algorithms are Las Vegas, and do not discuss any further issues of determinization henceforth.

\subsection{Technical Summary: Distributed Augmentation}
The results for forest-decomposition in multigraphs are based on augmenting paths, where we color one {\it uncolored} edge and possibly change some of the colored edges while maintaining solution feasibility. Augmentation approaches have been  used for many combinatorial constructions, such as coloring and matching.  The forest-decomposition algorithm of Gabow and Westermann \cite{GW92} also follows this approach.  Roughly speaking, it works as follows: given an uncolored edge $e_1$, we try to assign it color $c_1$. If no cycle is created, we are done. Otherwise, if it creates a cycle $C_1$, we recolor some edge $e_2$ on $C_1$ with a different color $c_2 \neq c_1$. Continuing this way gives an {\it augmenting sequence} $e_1, c_1, e_2,c_2 \ldots, e_{\ell},c_{\ell}$, such that recoloring $e_{\ell}$ in $c_{\ell}$ does not create a cycle.  This can be found using a BFS algorithm  in the centralized setting.

There are two main challenges for the \local model. First,  to get a distributed algorithm, we must color edges in parallel. Second, to get a local algorithm, we must restrict the recoloring to edges which are near the initial uncolored edge.  Note that the augmenting sequences produced by the Gabow-Westermann algorithm can be long and consecutive edges in the sequence  (e.g. $e_1$ and $e_2$) can be arbitrarily far  from each other.

\paragraph{Structural Results on Augmenting Sequences} We first show a structural result on forest decomposition: given a partial $(1+\eps)\arb$-FD (or, more generally, an LFD) in a multigraph, there is an augmenting sequence of length $O(\log n / \eps)$ where, moreover, every edge in the sequence lies in the $O(\log n / \eps)$-neighborhood of the starting uncolored edge $\einit$. This characterization may potentially lead to other algorithms for forest decompositions. We show this through a key modification to the BFS algorithm for finding an augmenting sequence. In \cite{GW92}, when assigning $e_i$ to color $c_i$ creates a cycle,  then all edges on the cycle get enqueued for the next layer; by contrast, in our algorithm, only the edges within distance $i$ of $\einit$ get enqueued. 

\paragraph{Network Decomposition and Removing Edges}  
We will parallelize the algorithm by breaking the graph into low-diameter subgraphs similar to \cite{GKM17}. However, there is a major roadblock we need to address: {\it identifying} an augmenting sequence may require checking edges distant from the uncolored edge. For example, edge $e_1$ may belong to a color-$c_1$ cycle which extends far beyond the vicinity of $e_1$.\footnote{A closely related computational model called \slocal was developed in \cite{GKM17}, where each vertex sequentially (in some order) reads its $r$-hop neighborhood for some radius $r$, and then produces its answer. If a problem has an \slocal algorithm with radius $r$, then it can be solved in $O(r \log^2 n)$-rounds in the \local model. Again, augmenting sequences need not lead to \slocal algorithms because of the need to check far-away edges.}

To sidestep this issue,  we develop a procedure \cut{} to remove edges, thereby breaking long paths and allowing augmenting sequences to be locally checkable. At the same time, we must ensure that the collection of edges removed by \cut{} (the ``left-over graph'') has arboricity $O(\eps \arb)$.  This can be viewed as an online load-balancing problem, where the load of a vertex is the number of directed neighbors which get removed. It is similar to a load-balancing problem encountered in \cite{SV19}, where paths come in an online fashion and we need to remove internal edges. Here, we encounter rooted trees instead of paths, and we need to remove edges to disconnect the root from all the leaves.

 If edges are removed independently, then the load of a vertex would be stuck at $\Omega(\log n)$ due to the concentration threshold.  To break this barrier, as in \cite{SV19}, we randomly remove edges incident to vertices with small load. We  show that throughout the algorithm, the root-leaf paths of the trees always contain many such vertices; thus, long paths are always killed with high probability.

\paragraph{Palette Partitioning for List-Coloring} The final step  is to recolor the left-over edges using an additional $O(\eps \arb)$ colors. For ordinary forest decomposition, this is nearly automatic due to our bound on the arboricity of the left-over graph. For list coloring, we must reserve a small number of back-up colors for the left-over edges. We develop two different  methods for this; the first uses the Lov\'{a}sz Local Lemma and the second uses randomized network decomposition.

\medskip
 
There are some additional connections in our work to two related graph parameters, \emph{pseudo-arboricity} and \emph{star-arboricity}. Let us summarize these next.

\subsection{Pseudo-Forest Decomposition and Low Outdegree Orientation}
There is a closely related decomposition using \emph{pseudo-forests}, which are graphs with at most one cycle in each connected component. The \emph{pseudo-arboricity} $\arb^*$ is the minimum number of pseudo-forests into which a graph can be decomposed.  A result of Hakimi \cite{hakimi} shows that pseudo-arboricity is given by an analogous formula to Nash-Williams' formula for arboricity, namely:
$$
\arb^{*}(G) = \max_{\substack{H \subseteq G \\ |V(H)| \geq 1}} \Bigg \lceil \frac{|E(H)|}{ |V(H)|} \Bigg \rceil~.
$$

In particular, as noted in \cite{PQ82}, loopless multigraphs have $\arb^* \leq \arb \leq 2 \arb^*$, and simple graphs have $\arb \leq \arb^* + 1$.

There is an equivalent, completely local, characterization of pseudo-arboricity: a \emph{$k$-orientation} of a graph is an orientation of the edges where every vertex has outdegree at most $k$. It turns out that $k = \arb^*$ is the minimum value for which such a $k$-orientation exists. In a sense,  $\arb^*$ is a more fundamental graph parameter than $\arb$, and the problems of pseudo-forest decomposition, low outdegree orientation, and maximum density subgraph are better-understood than forest decomposition. For example,  maximum density subgraph has been studied in many computational models, e.g. \cite{PQ82, Goldberg84, GGT89, Charikar00, KS09,BHNT15, EsfandiariHW15, McGregorTVV15, SJ20 ,BahmaniKV12, BahmaniGM14, GLM19 ,SarmaLNT12}. Low outdegree orientation has been studied in the centralized context in \cite{GW92, BF99, Kowalik06, GKS14, KKPS14, BB20}. 

There has been a long line of work on \local algorithms for $(1+\eps)\arb^{*}$-orientation \cite{GS17, FGK17, GHK18, Harris19, SV19b}.\footnote{For many of these works, the graph was implicitly assumed to be simple, and the algorithm provides a $(1+\eps) \arb$-orientation; since simple graphs have $\arb^* \leq \arb \leq \arb^* + 1$, this is a minor adjustment of the parameters. Also note that \cite{Harris19} claims a $(1+\eps) \arb$-orientation in multigraphs, but the algorithm actually provides a $(1+\eps) \arb^*$-orientation.}  Most recently, \cite{SV19b} gave an algorithm in $\tilde{O}(\log^2 n /\eps^2)$ rounds for $\eps \arb^* \geq 32$; this algorithm also works in the \congest model, which is a special case of the \local model where messages are restricted to $O(\log n)$ bits per round. 

Our general strategy of augmenting paths and network decompositions can also be used for low outdegree orientations. We will show the following result:

\begin{theorem}
\label{orient-cor}
For a (multi)-graph $G$ with pseudo-arboricity $\arb^*$ and $\eps \in (0,1)$, there is a \local algorithm to obtain $\lceil \arb^*(1+\eps) \rceil$-orientation in $O( \frac{\log^3 n}{\eps})$ rounds w.h.p.
\end{theorem}

Note in particular the linear dependency on $1/\eps$. For example, if $\arb^* = \sqrt{n}$, we can get an $(\arb^* + 1)$-orientation in $\tilde{O}(\sqrt{n})$ rounds, while previous  results would require $\Omega(n)$ rounds. Notably, the $1/\eps^2$ factor in \cite{SV19b} comes from the number of iterations needed to solve the LP. In addition to being a notable result on its own merits, Theorem~\ref{orient-cor} provides a simple warm-up exercise for our more advanced forest-decomposition algorithms.

\subsection{Star-Arboricity and List-Star-Arboricity for Simple Graphs}
The \emph{star-arboricity} $\arb_{\text{star}}$ is the minimum number of star-forests into which the edges of a graph can be partitioned.  This has been studied in combinatorics \cite{algor1989star,AOKI1990115,AMR92}. We analogously define $\arb_{\text{star}}^{\text{list}}$ as the smallest value $k$ such that an LSFD exists whenever each edge has a palette of size $k$ (this has not been studied before, to our knowledge).   For general loopless multigraphs,  it can be shown that $\arb_{\text{star}} \leq 2 \arb^*$ (see Proposition~\ref{star-pseudo-cycle}) and $\arb_{\text{star}}^{\text{list}} \leq 4\arb^*$ (see  Theorem~\ref{degen-thm}). In simple graphs, Alon, McDiarmid \& Reed \cite{AMR92} showed that $\arb_{\text{star}} \leq \arb + O(\log \Delta)$. 

Our algorithms for star-forest-decomposition in simple graphs come from a strengthened version of the construction of \cite{AMR92}.   To briefly summarize, consider some fixed $k$-orientation of the graph. For each color $c$, mark each vertex as a $c$-center independently with some probability $p$. Then finding a star-forest decomposition reduces to finding a perfect matching, for each vertex $u$, between the colors $c$ for which $u$ is a not a $c$-center, and the out-neighbors of $u$ which are $c$-centers.

In the general LSFD case, we show that these perfect matchings exist, and can be found efficiently, when $\eps k = \Omega(\log \Delta)$. This is based on more advanced analysis of concentration bounds for the number of $c$-leaf neighbors.  In the ordinary SFD case, instead of perfect matchings, we obtain near-perfect matchings, leaving $\eps k$ unmatched edges per vertex. These left-over edges can later be  decomposed into $2\eps k$ stars. This gives a bound $\eps k = \Omega(\sqrt{\log \Delta} + \log \arb)$.

 In addition to being powerful algorithmic results, these also give two new combinatorial bounds:
\begin{corollary}
\label{simple-res1}
A simple graph has  $\arb_{\text{star}} \leq \arb + O({\log \arb} + \sqrt{\log \Delta})$ and $\arb_{\text{star}}^{\text{list}} \leq \arb + O(\log \Delta)$.
\end{corollary}

For lower bounds, \cite{AMR92} showed that there are simple graphs with $\arb_{\text{star}} = 2 \arb$ and $\Delta = 2^{\Theta(\arb^2)}$, while  \cite{algor1989star} showed that there are simple graphs where every vertex has degree exactly $d = 2 \arb$ and where $\arb_{\text{star}} \geq \arb + \Omega(\log \arb)$.   These two lower bounds  show that the dependence of $\arb_{\text{star}}$ on $\arb$ and $\Delta$ are nearly optimal in Corollary~\ref{simple-res1}.   In particular, the term $\log \arb$ cannot be replaced by a function $o(\log \arb)$ and the term $\sqrt{\log \Delta}$ cannot be replaced by a function  $o(\sqrt{\log \Delta})$.


\subsection{Preliminaries}
\label{prelim-sec}
Our algorithms will frequently use global parameters such as $\eps, \arb, \arb^*, m, n,  \Delta$. As usual in distributed algorithms, we always suppose that we are given some globally-known upper bounds on such values as part of the input;  when we write $\arb, n$ etc. we are technically referring to input values $\hat \arb, \hat n$ etc. which are upper bounds on them. Almost all of our results become vacuous if $\eps < 1/n$ (since, in the \local model, we can simply read in the entire graph in $O(n)$ rounds), so we assume throughout that $\eps \in (1/n, 1/2)$.

We define the $r$-neighborhood of a vertex $v$, denoted $N^r(v)$, to be the set of vertices within distance $r$ of $v$. We likewise write $N^r(e)$ for an edge $e$ and $N^r(X)$ for a set $X$ of vertices or edges. For any vertex set $X$, we define $E(X)$ to be the set of induced edges on $X$. We define the power-graph $G^r$ to be the graph on vertex set $V$ and with an edge $uv$ if $u, v$ have distance at most $r$ in $G$. Note that, in the \local model, $G^r$ can be simulated in $O(r)$ rounds of $G$.

For any integer $t \geq 0$, we define $[t] = \{1, \dots, t \}$. We write $A = B \sqcup C$ for a disjoint union, i.e. $A = B \cup C$ and $B \cap C = \emptyset$.

\paragraph{Concentration bounds} At several places, we refer to Chernoff bounds on sums of random variables. To simplify formulas, we define $F_+(\mu, t) = \bigl( \frac{e^{\delta}}{(1+\delta)^{\delta}} \bigr)^{\mu}$ where $\delta = t/\mu - 1$, i.e. the upper bound on the probability that a Binomial random variable with mean $\mu$ takes a value as large as $t$.   Some well-known bounds are $F_+(\mu,  t) \leq (e \mu/t)^t$ for any value $t \geq 0$, and $F_+(\mu, \mu (1+\delta)) \leq e^{-\mu \delta^2/3}$ for $\delta \in [0,1]$.  Chernoff bounds also apply to certain types of negatively-correlated random variables; for instance, we have the following standard result  (see e.g.~\cite{PS97}):
\begin{lemma}
\label{chernoff-ext}
Suppose that $X_1, \dots, X_k$ are Bernoulli random variables and for every $S \subseteq [k]$ it holds that $\Pr(\bigwedge_{i \in S} X_i = 1) \leq q^{|S|}$ for some parameter $q \in [0,1]$. Then, for any $t \geq 0$, we have $\Pr( \sum_{i=1}^k X_i \geq t ) \leq F_+(k q, t)$.
\end{lemma}

\paragraph{Lov\'{a}sz Local Lemma (LLL)} The LLL is a general principle in probability theory which states that for  a collection of ``bad'' events  $\mathcal B = \{ B_1, \dots, B_t \}$ in a probability space, where each event has low probability and is independent of most of the other events, there is a positive probability that none of the events $B_i$ occur. It often appears in the context of graph theory and distributed algorithms, wherein each bad-event is some locally-checkable property on the vertices. 

We will use a randomized \local algorithm of \cite{chung2017distributed}  to determine values for the variables to avoid the bad-events. This algorithm runs in $O(\log n)$ rounds under the criterion $e p d^2 \leq 1 - \Omega(1)$, where   $p$ is the maximum probability of any bad-event and $d$ is the maximum number of bad-events $B_j$ dependent with any given $B_i$,    Note that this is stronger than the general symmetric LLL criterion, which requires merely $e p d \leq 1$.

\paragraph{Network Decomposition} A $(D,\chi)$-network decomposition is a partition of the vertices into $\chi$ classes such that every connected component in every class has strong diameter at most $D$. Each connected component within each class is called a  \emph{cluster}. An $( O(\log n), O(\log n) ) $-network decomposition can be obtained in $O(\log^2 n)$ rounds by randomized \local algorithms \cite{LinialS93, ABCP96, EN16}.

We also consider a related notion of \emph{$(D, \beta)$-stochastic network decomposition}, which is a randomized procedure to select an edge set $L \subseteq E$ such that (i) the connected components of the graph $G' = (V, L)$ have strong diameter at most $D$, and (ii) any given edge goes into $L$ with probability at least $1 - \beta$. There is an algorithm of \cite{miller2013parallel} to produce a $(O(\frac{\log n}{\beta}), \beta)$-stochastic network decomposition in $O( \frac{\log n}{\beta})$ rounds of the \local model.


\subsection{Basic Forest Decomposition Algorithms}
\label{crude-sec}
We list here some simpler algorithms for forest decomposition.  These will be important building blocks for later and may also be of independent combinatorial and algorithmic interest.  
\begin{theorem}
\label{ace1prop}
Let $t = \lfloor (2 + \eps) \arb^* \rfloor$ for $\eps \in (0,1)$. There are deterministic $O( \frac{\log n}{\eps})$-round algorithms to obtain the following decompositions of $G$:
\begin{itemize}
\item A partition of the vertices into $k = O(\frac{\log n}{\eps})$ classes $H_1, \dots, H_k$, such that each vertex $v \in H_i$ has at most $t$ neighbors in $H_i \cup \dots \cup H_k$.
\item An orientation of the edges such that the resulting directed graph is acyclic and each vertex has outdegree at most $t$. (We refer to this as an \emph{acyclic $t$-orientation}).
\item A $3t$-star-forest-decomposition.
\item A list-forest-decomposition when every edge has a palette of size $t$.
\end{itemize}
\end{theorem}

\begin{theorem}
\label{lem:Hpartition3}
If every edge has a palette of size $\lfloor (4+\eps) \arb^* \rfloor$ for $\eps \in (0,1)$, then an LSFD can be computed in $\min \bigl \{ O(\frac{\log^3 n}{\eps}), \tilde O( \frac{\log n \log^*m}{\eps}) \bigr \}$ rounds w.h.p.
\end{theorem}

\begin{proposition} 
\label{decompose-prop3}
Suppose we are given some $k$-FD of the multigraph $G$, of arbitrary diameter. For any $\eps > 0$, there is an $O(\frac{\log n}{\eps})$-round algorithm to compute a $(k+\lceil \eps \arb \rceil)$-FD of diameter $D \leq O( \frac{\log n}{\eps} )$ w.h.p. If $\arb \geq \Omega \bigl( \min \{ \frac{\log n}{\eps}, \frac{\log \Delta}{\eps^2} \} \bigr)$, we can get $D \leq O( \frac{1}{ \eps})$ w.h.p. with the same runtime.
\end{proposition}

The first two results of Theorem~\ref{ace1prop} were shown (with slightly different terminology) in the $H$-partition algorithm of \cite{BE10}; we also provide full proofs in Appendix~\ref{ace1propapp}.  The proof of Theorem~\ref{lem:Hpartition3} appears in Appendix~\ref{Hpart3-sec}. The proof of Proposition~\ref{decompose-prop3}  appears in Appendix~\ref{dp3-app}, along with a more general result we will need later for reducing the diameter of list forest decompositions.

\section{Algorithm for Low Outdegree Orientation}
\label{sec:outdegree}
We now discuss a \local algorithm for $(1+\eps)\arb^{*}$-orientation, based on augmenting sequences and network decomposition.  Consider a multigraph $G$ of pseudo-arboricity $\arb^*$.  \textbf{For the purposes of this section only, we allow $G$ to contain loops.} Following \cite{GS17}, we can ``augment'' a given edge-orientation $\psi$ by reversing the edges in some directed path.  We begin with the following observation, which is essentially a reformulation of results of \cite{hakimi, GS17} with more careful counting of parameters.
\begin{lemma}
Let $\eps \in (0,1)$. For any edge-orientation $\psi$ and vertex $v$, there is a directed path of length $O( \frac{\log n}{\eps} )$ from $v$ to a vertex with outdegree strictly less than $a^*(1+\eps)$.
\end{lemma}
\begin{proof}
For $i \geq 0$, let $V_i$ denote the vertices at distance at most $i$ from $v$. If all vertices in $V_i$ have outdegree at least $a^*(1+\eps)$, then for each $j < i$ we can count the edge-set $E(V_{j+1})$ in two ways. First, each vertex in $V_{j}$ has outdegree at least $a^*(1+\eps)$ and these edges have both endpoints in $V_{j+1}$, so $|E(V_{j+1})| \geq |V_{j}| a^* (1+\eps)$. On the other hand, by definition of pseudo-arboricity, we have $|E(V_{j+1})| \leq |V_{j+1}| a^*$. Putting these two observations together, we see that
\[
|V_{j+1}| \geq |V_{j}| \cdot (1+\eps) \qquad \text{for $j = 0, \dots, i-1$}.
\]

Since $V_0 = \{v \}$, by telescoping products this implies that $|V_i| \geq (1+\eps)^i$. Since clearly $|V_i| \leq n$, we must have $i \leq \log_{1+\eps} n \leq O( \frac{\log n}{\eps} )$.
\end{proof}

For the purposes of our algorithm, the main significance of this result is that we can locally fix a given edge-orientation.  For a given parameter $\eps > 0$, let us say that a vertex $v$ is \emph{overloaded} with respect to a given orientation $\psi$, if the outdegree of $v$ is strictly larger than $\lceil \arb^*(1+\eps) \rceil$; otherwise, if the outdegree is at most $\lceil \arb^*(1+\eps) \rceil$, it is \emph{underloaded}.
 We summarize this as follows:
\begin{proposition}
\label{aug:cor1}
Suppose multigraph $G$ has an edge-orientation $\psi$, and let $U \subseteq V$ be an arbitrary vertex set. Then there is an edge-orientation $\psi'$ with the following properties:
\begin{itemize}
\item $\psi'$ agrees with $\psi$ outside $N^r(U)$ where $r = O( \frac{\log n}{\eps})$.
\item All vertices of $U$ are underloaded with respect to $\psi'$.
\item All vertices which are underloaded with respect to $\psi$ remain underloaded with respect to $\psi'$. 
\end{itemize}
\end{proposition}
\begin{proof}
Following \cite{GS17}, consider the following process: while some vertex of $U$ is overloaded, we choose any arbitrary such vertex $v \in U$. We then find a shortest directed path $v, v_1, \dots, v_r$ from $v$ where vertex $v_r$ has outdegree strictly less than $a^*(1+\eps)$. Next, we reverse the orientation of all edges along this path. This does not change the outdegree of the vertices $v_1, \dots, v_{r-1}$, while it decreases the outdegree of $v$ by one and increases the outdegree of $v_r$ by one.  

This process never creates a new overloaded vertex,  while decreasing the outdegree of $v$ at each step. Thus, after finite number of steps, it terminates and all the vertices in $U$ are underloaded. Also, each step of this process only modifies edges within distance $r \leq O(\frac{\log n}{\eps})$ of some vertex  of $U$. Hence, $\psi'$ agrees with $\psi$ for all other edges.
\end{proof}

We remark that this type of ``local patching'' has also been critical for  other \local algorithms, such as the $\Delta$-vertex-coloring algorithm of \cite{ghaffari2018improved} or the $(\Delta+1)$-edge-coloring algorithm of \cite{bernshteyn2020fast}. We next use network decomposition to extend this local patching into a global solution, via the following Algorithm~\ref{alg:main0}. Here, $K$ is a universal constant to be specified.

 \begin{algorithm}[H]
\caption{\textsc{Low-degree\_Orientation\_Decomposition}$(G)$}\label{alg:main0}
\begin{algorithmic}[1]\small

\State Initialize $\psi$ to be some arbitrary orientation of $G$.
\State Compute an $( O(\log n), O(\log n) )$-network decomposition in $G^{2 R}$ for $R = \lceil K \log n / \eps \rceil$.

\For{\textbf{each}  class $z$ in the network decomposition} 

		\For{each component $U$ in the class $z$ {\bf in parallel}}

			\State Modify $\psi$ so that vertices inside $U$ become underloaded,  vertices outside $N^R(U)$ are unchanged, and no new overloaded vertices are created.
			 		\EndFor
\EndFor
\end{algorithmic}
\end{algorithm}

\begin{theorem}
Algorithm~\ref{alg:main0} runs in  $O(\log^3 n / \eps)$ rounds. At the termination, the edge-orientation $\psi$ has maximum outdegree $\lceil a^*(1+\eps) \rceil$ w.h.p.
\end{theorem}
\begin{proof}
For Line 2, we use the algorithm of \cite{EN16} to obtain the network decomposition for $G^{2 R}$ in $O( R \log^2 n)$ rounds.  Algorithm~\ref{alg:main0} processes each cluster $U$ of a given class simultaneously, and we also define $U' = N^R(U)$. From Proposition~\ref{aug:cor1}, it is possible to modify $\psi$ within $U'$ for sufficiently large $K$, such that all vertices in $U$ become underloaded, and no additional overloaded vertices are created. This can be done by having some ``leader'' vertex in each cluster $U$ read in the neighborhood $N^R(U)$ and choose a modified $\psi'$ and broadcast it to the other nodes in the cluster.
 
The distance between two clusters in the same class is at least $2 R + 1$. Moreover,  if $u,v$ are adjacent in $G^{2 R}$, their distance in $G$ is at most $2 R$. So each cluster $U$ has weak diameter at most $O(R \log n)$, and also the balls $U'_1$ and $U'_2$ must be disjoint for any two clusters  $U_1$ and $U_2$ of the same class.  Therefore, each iteration can be simulated locally in $ O( R \log n)$ rounds.  Since there are $O(\log n)$ classes, the total running time is $O(R \log^2 n) = O(\log^3 n / \eps)$.
\end{proof}

Algorithm~\ref{alg:main0} can be viewed as part of a family of algorithms based on network decomposition described in \cite{GKM17}. (In the language of \cite{GKM17},  the algorithm can be implemented in \slocal with radius $r = O(\frac{\log n}{\eps})$.) However, we describe the algorithm explicitly to keep this paper self-contained, and because we later need a more general version of Algorithm~\ref{alg:main0}.

We will use the same overall strategy for forest decomposition, but we will encounter two technical obstacles. First, we must define an appropriate notion of local patching and augmenting sequences; this will be far more complex than Proposition~\ref{aug:cor1}.  Second, and more seriously, forest-decomposition, unlike low-degree orientation, cannot be locally checked due to the possibility of long cycles. To circumvent this, we must remove edges at each step to destroy these cycles. These leftover edges will need some post-processing steps at the end.

\section{Augmenting Sequences for List-Coloring}
\label{sec:aug}
\sloppy
We now show our main structural result on augmenting sequences. Given a partial LFD $\psi$ of multigraph $G$ and an edge $e=uv$, we define $C(e, c)$ to be the unique $u$-$v$ path in the $c$-colored forest; if $u$ and $v$ are disconnected in the color-$c$ forest then we write $C(e,c) = \emptyset$.

We define an \emph{augmenting sequence}  w.r.t.~$\psi$ to be a sequence $P = (e_1, e_2, \dots,  e_{\ell}, c)$, for edges $e_i$ and color $c$,  satisfying the following four conditions:

\begin{itemize}
\item[(A1)] $e_{i} \in C(e_{i-1}, \psi(e_i)  )$ for $2 \leq i \leq \ell$.
\item[(A2)] $e_{i} \notin C(e_{j}, \psi(e_i))$ for every $i$ and $j$ such that $j < i-1$.
\item[(A3)] $C(e_{\ell}, c) = \emptyset$.
\item[(A4)] $\psi(e_{i+1}) \in Q(e_i)$ for each $i = 1, \dots, \ell-1$ and $c \in Q(e_{\ell})$.
\end{itemize}

 Recall that $Q(e)$ denotes the list of available colors for edge $e$.   We say that $\ell$ is the \emph{length} of the sequence.  We define the \emph{augmentation} $\psi'$ to be a new (partial) coloring obtained by setting $\psi'(e_i) = \psi(e_{i+1})$ for $1 \leq i \leq \ell-1$ and $\psi'(e_{\ell}) = c$, and $\psi'(e) = \psi(e)$ for all other edges $e \in E \setminus \{e_1, \ldots, e_{\ell}\}$. See Figure~\ref{aug_sequence_1}.

\begin{figure}[h]
\centering
\begin{subfigure}[b]{0.28\textwidth}
\centering
\includegraphics[scale=1]{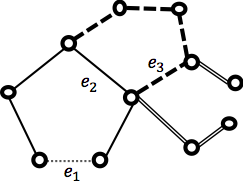}
\end{subfigure}
\qquad
\begin{subfigure}[b]{0.28 \textwidth}
\centering
\includegraphics[scale=1]{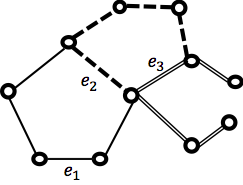}
\end{subfigure}
\caption{An illustration of an augmenting sequence before and after the augmentation process. \label{aug_sequence_1}}
\end{figure}

\begin{lemma} \label{lem:augment} For an augmenting sequence $P$ w.r.t. $\psi$, the augmentation  $\psi'$ remains a partial list-forest decomposition.
 \end{lemma}

\begin{proof}
In this proof, the notation $C(e, b)$ always refer to the cycles with respect to the original coloring $\psi$.

We first claim that $\psi'(e_i) \neq \psi'(e_{i+1})$ for $i = 1, \dots, \ell - 1$. For $i = \ell - 1$, this follows from (A3) since $C(e_{\ell}, \psi'(e_{\ell})) = C(e_{\ell}, c) = \emptyset$ while $C(e_{\ell}, \psi'(e_{\ell-1})) = C(e_{\ell}, \psi(c_{\ell})) = \{ e_{\ell} \}$. For $i < \ell - 1$, suppose for contradiction that  $\psi'(e_i) = \psi'(e_{i+1})$, i.e. $\psi(e_{i+1}) = \psi(e_{i+2})$.  By (A1) applied at $i+2$, we then have $e_{i+2} \in C(e_{i+1}, \psi(e_{i+2})) = C(e_{i+1}, \psi(e_{i+1})) = \{e_{i+1} \}$. So $e_{i+2} = e_{i+1}$. But then (A1) applied at $i+1$ would give $e_{i+2} = e_{i+1} \in C(e_i,  \psi(e_{i+1}))$, which contradicts (A2).

Now for each $i = 1, \dots, \ell+1$, define $\phi_i$ to be the coloring obtained by setting $\phi_i(e_j) = \psi'(e_j)$ for all $j = i, \dots, \ell$, and $\phi_i(e) = \psi(e)$ for all other edges $e$.  Thus, $\psi = \phi_{\ell+1}$ and $\psi' = \phi_1$. 

By (A3), $\phi_{\ell}$ does not have a cycle. So if $\psi'$ is not a partial LFD,  let $i \in \{1, \dots, \ell - 1 \}$ be maximal such that $\phi_i$ has a cycle. Since $\phi_{i}$ and $\phi_{i+1}$ differ only at edge $e_i$, it must be that $\phi_{i+1}$ has a path $p_1$ on color $\psi'(e_{i})$ from $u$ to $v$.  Since $\phi_{i+1}(e_{i+1}) = \psi'(e_{i+1}) \neq \psi'(e_i)$,  path $p_1$ does not contain edge $e_{i+1}$. Since $\phi_{i+2}$ and $\phi_{i+1}$ only differ at edge $e_{i+1}$ and $e_{i+1} \notin p_1$, path $p_1$ is also present in $\phi_{i+2}$. On the other hand, by (A1), we have $e_{i+1} \in p_2$, for path $p_2 = C(e_i, \psi(e_{i+1}))$. By (A2), none of the edges $e_{i+2}, \dots, e_{\ell}$ were on $p_2$, hence  $p_2$ remains in $\phi_{i+2}$. 

We must have $p_1 \neq p_2$ since $e_{i+1} \notin p_1$ but $e_{i+1} \in p_2$. Thus $\phi_{i+2}$ contains two distinct paths of the same color from $u$ to $v$. This contradicts the maximality of $i$.
\end{proof}

With this definition, we will show the following main result:
\begin{theorem} \label{thm:augmenting-length}Given a partial LFD of a multigraph $G$ where every edge has palette size $(1+\eps) \arb$, and an uncolored edge $e$, there is an augmenting sequence $P = (e, e_2, \dots, e_{\ell}, c)$ from $e$ where $e_2, \dots, e_{\ell} \in N^{r}(e)$ for $r = O( \frac{\log n}{\eps})$.
 \end{theorem}

The main significance of Theorem~\ref{thm:augmenting-length} is that it allows us to locally fix a partial LFD, in the same way Proposition~\ref{aug:cor1} allows us to locally fix an edge-orientation. We summarize this as follows:
\begin{corollary}
\label{aug:cor}
Suppose multigraph $G$ has a partial LFD $\psi$, and every edge has palette size $(1+\eps) \arb$, and let $L \subseteq E$ be an arbitrary edge set. Then there is a partial LFD $\psi'$ with the following properties:
\begin{itemize}
\item $\psi'$ agrees with $\psi$ outside $N^r(L)$ where $r = O( \frac{\log n}{\eps})$.
\item $\psi'$ is a full coloring of the edges $L$.
\item All edges colored in $\psi$ are also colored in $\psi'$.
\end{itemize}
\end{corollary}
\begin{proof}[Proof (assuming Theorem~\ref{thm:augmenting-length})]
We can go through each uncolored edge $e \in L$ in an arbitrary order, obtain an augmenting sequence $P$ from Theorem~\ref{thm:augmenting-length}, and then replace $\psi$ with its augmentation w.r.t $P$. This ensures that $e$ is colored, and does not de-color any edges. Furthermore, since $P$ lies inside $N^r(e)$, it does not modify any edges outside $N^r(L)$. At the end of this process, all edges in $L$ have become colored, and none of the edges outside $N^r(L)$ have been modified.
\end{proof}

To prove Theorem~\ref{thm:augmenting-length}, we first construct a weaker object called an \emph{almost augmenting sequence}, which is a sequence satisfying properties  (A1), (A3), (A4) but not necessarily (A2).   The following Algorithm~\ref{alg:find_augmenting_sequence} finds an almost augmenting sequence starting from a given edge $e_{\text{init}}$. 

 \begin{algorithm}[H]
\caption{\textsc{Find\_Almost\_Augmenting\_sequence}$(e_{\text{init}})$}\label{alg:find_augmenting_sequence}
\begin{algorithmic}[1]\small

\State $E_1 = \{e_{\text{init}}\}$
\For{$i = 1 \ldots n$} \label{ln:loop1}
	\State $E_{i+1} \leftarrow E_{i}$.
	\For{\textbf{each} $e \in E_{i}$ \textbf{and each} color $c \in Q(e)$}
			\If{$C(e,c) \neq \emptyset$}
				\For{\textbf{each} edge $e' \in C(e,c) \setminus E_i$ which is adjacent to an edge of $E_i$} 
				\State Set $E_{i+1} \leftarrow E_{i+1} \cup \{e' \}$ and $\pi(e') \leftarrow e$.
				\EndFor
			\Else
				\State Return the almost augmenting sequence $P = (e_{\text{init}}, \dots, \pi(\pi(e)), \pi(e), e, c)$.
			\EndIf			
		\EndFor
\EndFor
\end{algorithmic}
\end{algorithm}

\begin{lemma}
\label{prop:find}
Algorithm~\ref{alg:find_augmenting_sequence} terminates within $O( \frac{\log n}{\eps})$ iterations.
\end{lemma}
\begin{proof}
In each iteration $i$, let $V_i$ denote the endpoints of the edges in $E_i$, and let $E_{i,c}$ be the set of edges in $E_i$ whose palette contains color $c$. An edge only gets added to $E_{i+1}$ if it is adjacent to an edge in $E_i$. Thus, the graph spanned by $(V_i, E_{i})$ is connected and $E_{i} \subseteq N^{i-1}(e_{\text{init}})$.

Let us assume we are at some iteration $i$ and the algorithm has not terminated, i.e. $C(e,c) \neq \emptyset$ for all $e \in E_{i,c}$.  For each color $c$, let $r_c$ be the number of connected components in the subgraph $G_{c} = (V_i, E_{i,c})$.  Consider forming a graph $H_c$ on vertex set $V$ with edge set given by $E(H_c) =  \bigcup_{e \in E_{i,c}}C(e,c)$. This is a forest consisting of $c$-colored edges. By our assumption that $C(e,c) \neq \emptyset$ for all $e \in E_{i,c}$, any vertices in $V_i$ which are connected in $G_c$ are also connected in $H_c$. Thus, $H_c$ has at most $r_c$ components, and if we choose some arbitrary rooting of forest $H_c$ then at most $r_c$ vertices in $V_i$ can be root nodes. 

Now, consider any such non-root node $x \in V_i$, with parent edge $e' \in H_c$. We have $e' \in C(e,c)$ for some edge $e \in E_{i,c}$. Since $x$ is an endpoint of $e'$ and is also an endpoint of an edge in $E_i$, the new edge $e'$ gets added to $E_{i+1}$ at iteration $i$, unless it was already part of $E_i$.  This holds for every non-root node in $H_c$, so $E_{i+1}$ contains at least $|V_i| - r_c$ edges from $H_c$, which are $c$-colored. (See Figure~\ref{fig:aug}.)  We sum over colors $c$ to get:
\begin{align*}
|E_{i+1}| &\geq  \sum_{c \in \mathcal{C}} \bigl( |V_i| - r_c \bigr).
\end{align*}

To bound this sum, consider an arbitrary spanning tree $T$ of the connected graph $(V_i, E_i)$, where $|T| = |V_i| - 1$. Since $T$ is a tree, we have $|T \cap E_{i,c} | \leq |V_i| - r_c$ for each color $c$, and so:
$$
|E_{i+1}| \geq \sum_{c \in \mathcal{C}} (|V_i| - r_c ) \geq \sum_{c \in \mathcal{C}} |T \cap E_{i,c}|  =  \sum_{e \in T} |Q(e)| \geq |T| \cdot (1+\eps) \arb = (1+\eps)\arb (|V_i| - 1)~. 
$$

Since $|V_1| = 2$, this implies $|E_2| \geq (1+\eps) \arb$. For iteration $i > 1$, note that by definition of arboricity, we have $|E_i|/(|V_i| - 1) \leq \arb$, and so 
$$
|E_{i+1}| \geq (1+\eps) \arb \cdot |E_i| / \arb = (1+\eps)  |E_i|.
$$

Hence $|E_{\ell+1}| \geq (1+\eps)^{\ell} \arb$ for each $\ell \geq 1$. The overall graph has $m \leq n \arb$ edges, so the process must terminate by iteration $\ell = \lceil \log_{1+\eps} n \rceil$.
\end{proof}

\begin{figure}[H]
\centering
\begin{subfigure}[b]{0.45\textwidth}
\centering
\includegraphics[scale=0.3]{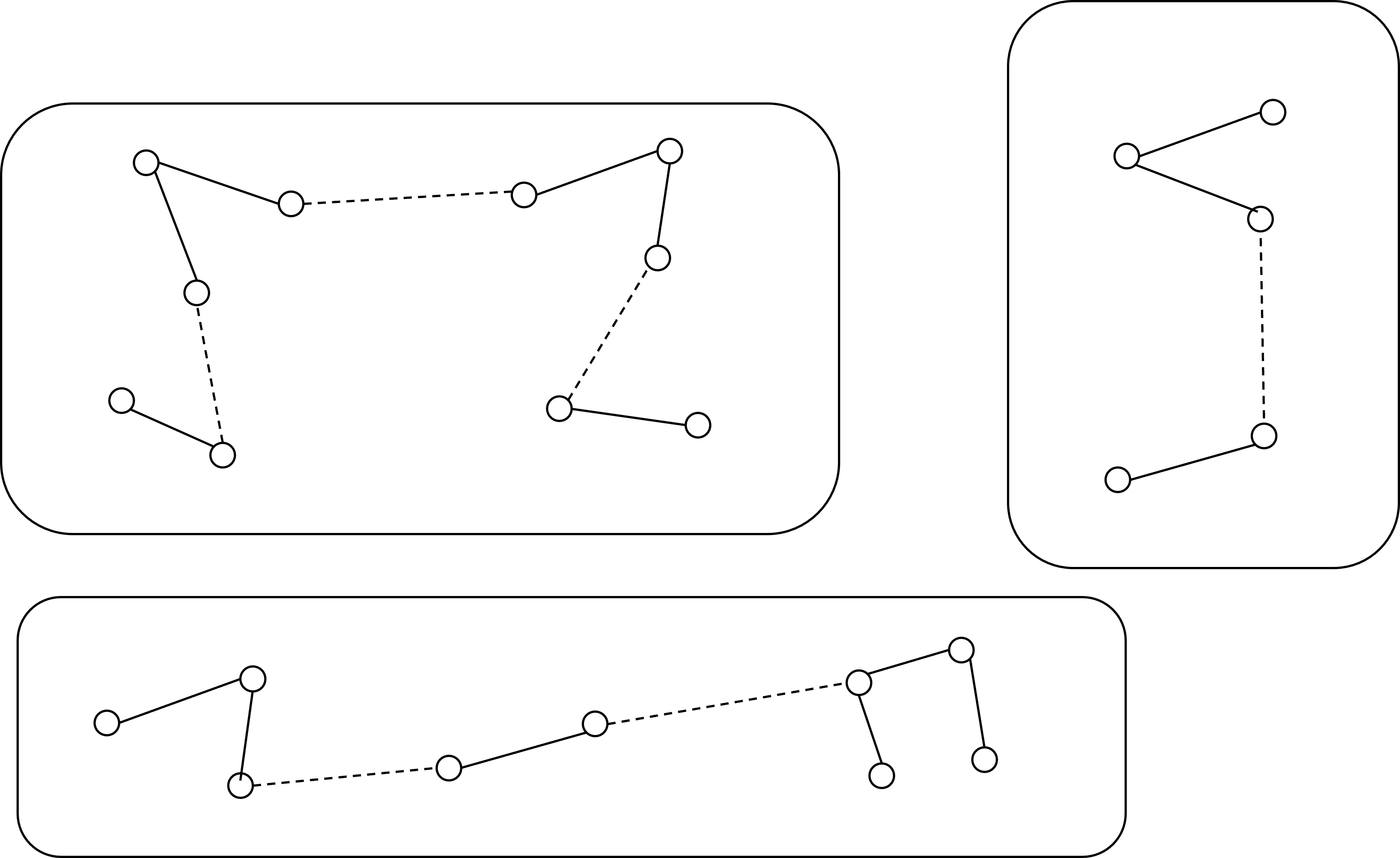}
\end{subfigure}
\quad
\begin{subfigure}[b]{0.45\textwidth}
\centering
\includegraphics[scale=0.3]{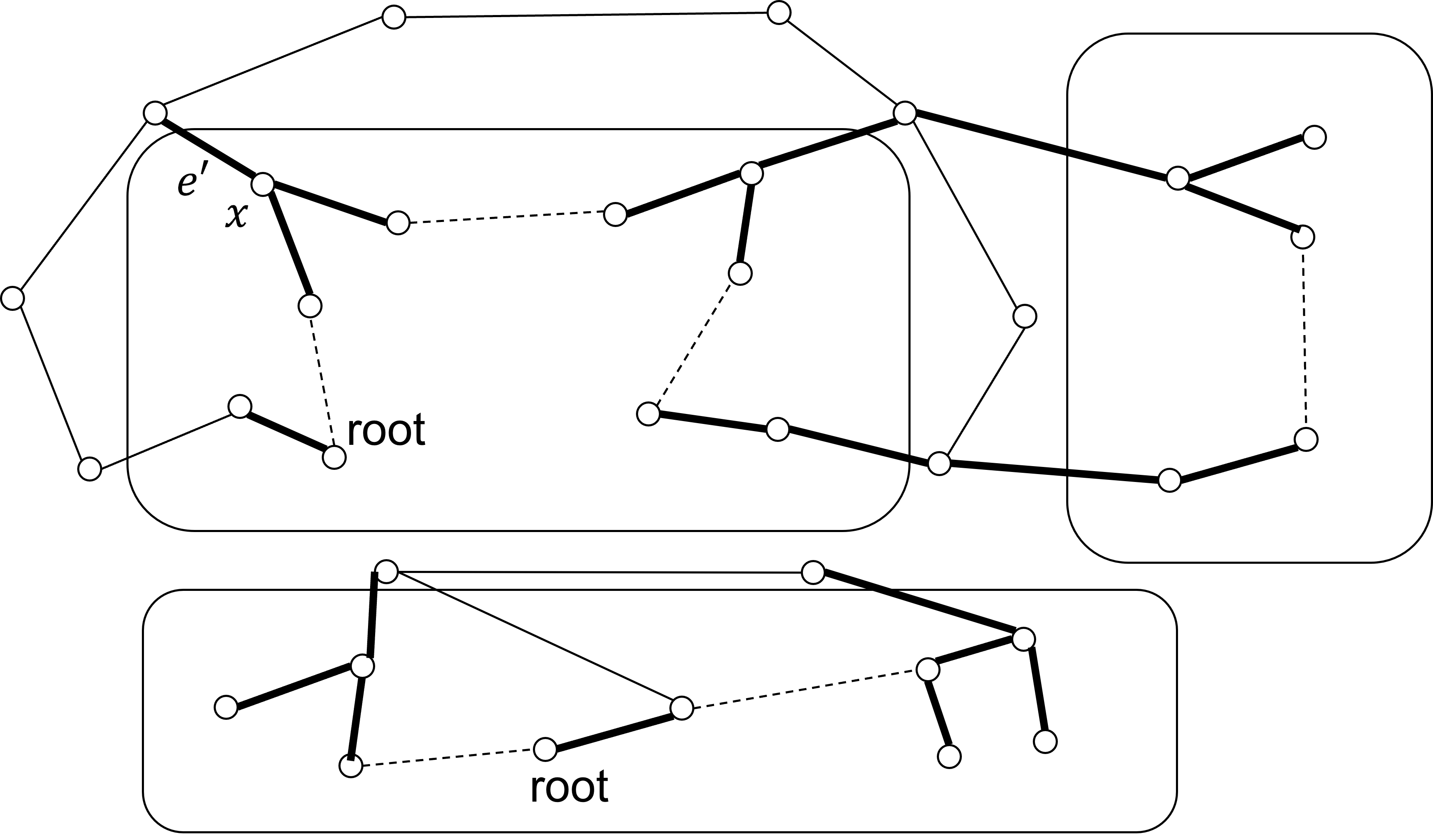}
\end{subfigure}
\caption{In the left-most figure, the ovals represent the connected component of $G_c$. The right-most figure shows the graph $H_c$; here, the bold edges are guaranteed to be added to $E_{i+1}$.\label{fig:aug}}
\end{figure}

 Note that if Algorithm~\ref{alg:find_augmenting_sequence} terminates by iteration $\ell = O( \frac{\log n}{\eps})$, then the sequence has length $\ell$ and all edges are within distance $\ell$ of the starting edge $e_{\text{init}}$. We can then short-circuit it into an augmenting sequence as shown in the following result:
\begin{proposition}\label{lem:shortcut} If there exists an almost augmenting sequence $P$  from $e$ to $e'$, then there exists an augmenting sequence from $e$ to $e'$ which is a subsequence of $P$. \end{proposition}
\begin{proof}
We show this by induction on $\ell$. If $\ell = 0$ it holds vacuously. Otherwise, consider an almost augmenting sequence $P = (e_1, e_2, \ldots e_{\ell}, c)$ with $e_1 = e$ and $e_{\ell} = e'$. If $P$ satisfies (A2) we are done. If not, suppose that $e_i \in C(e_j, \psi(e_i))$ for $j < i - 1$. Then, $P' = (e_1, \ldots, e_j,  e_i,  \ldots e_{\ell}, c)$ is also an almost augmenting sequence of length $\ell' < \ell$ which is a subsequence of $P$. By induction hypothesis, it has a subsequence $P''$ which is augmenting sequence, which in turn is a subsequence of $P$.
 \end{proof}

Theorem~\ref{thm:augmenting-length} now follows immediately from Lemma~\ref{prop:find} and Proposition~\ref{lem:shortcut}.

\section{Local Forest Decompositions via Augmentation} \label{sec:main-alg}


Algorithm \ref{alg:main} is a high-level description of our forest decomposition algorithm, in terms of a parameter $R$,  a constant $K'$, and a subroutine $\cut$ (all to be specified).
 \begin{algorithm}[H]
\caption{\textsc{Forest\_Decomposition}$(G)$}\label{alg:main}
\begin{algorithmic}[1]\small

\State Initialize $\psi \leftarrow \emptyset$.
\State Compute an $( O(\log n), O(\log n) )$-network decomposition in $G^{2\cdot (R + R')}$ for $R' = \lceil K' \log n / \eps \rceil$. \label{step:nd}

\For{\textbf{each}  class $z$ in the network decomposition} \label{step:color-class}

		\For{each component $U$ in the class $z$ {\bf in parallel}}  \label{ln:outerloopstart}

				\State  Execute \cut{$(U,R)$}.
			
			\State Modify $\psi$ within $N^{R'}(U)$ so that edges inside $N(U)$ become colored.
			 \label{ln:lastline}
		\EndFor
\EndFor
\end{algorithmic}
\end{algorithm}

The subroutine $\cut(U,R)$ removes edges from the graph so as to break all long monochromatic paths in the vicinity of $U$.  We call the set of removed edges from all $\cut(U,R)$ instances the {\it leftover edges},  denote by by $E^{\text{leftover}}$; the graph induced on them is the \emph{leftover graph} denoted $G^{\text{leftover}}$. We also define the \emph{main edges} by $E^{\text{main}} = E \setminus E^{\text{leftover}}$ (i.e. the edges never removed by $\cut$) and the induced graph on them the \emph{main graph} $G^{\text{main}}$.

Formally, for each class $U$, define $U' = N^{R'}(U)$ and $U'' = N^{R+R'}(U)$, and define $H(U)$ to be the graph on the edges in $U'' \setminus U'$. Furthermore, for each color $c \in \mathfrak C$, define $H_c(U)$ to be the $c$-colored edges in $H(U)$; note that $H_c(U)$ is a forest. We say that the execution of Algorithm~\ref{alg:main} \emph{is good} if, after every application of $\cut(U,R)$, the vertex sets $U'$ and $V \setminus U''$ are disconnected in $H_c(U)$ for every color $c$. (See Figure~\ref{fig:cut_2}.)  

\begin{figure}[H]
\centering
\includegraphics[scale=0.54]{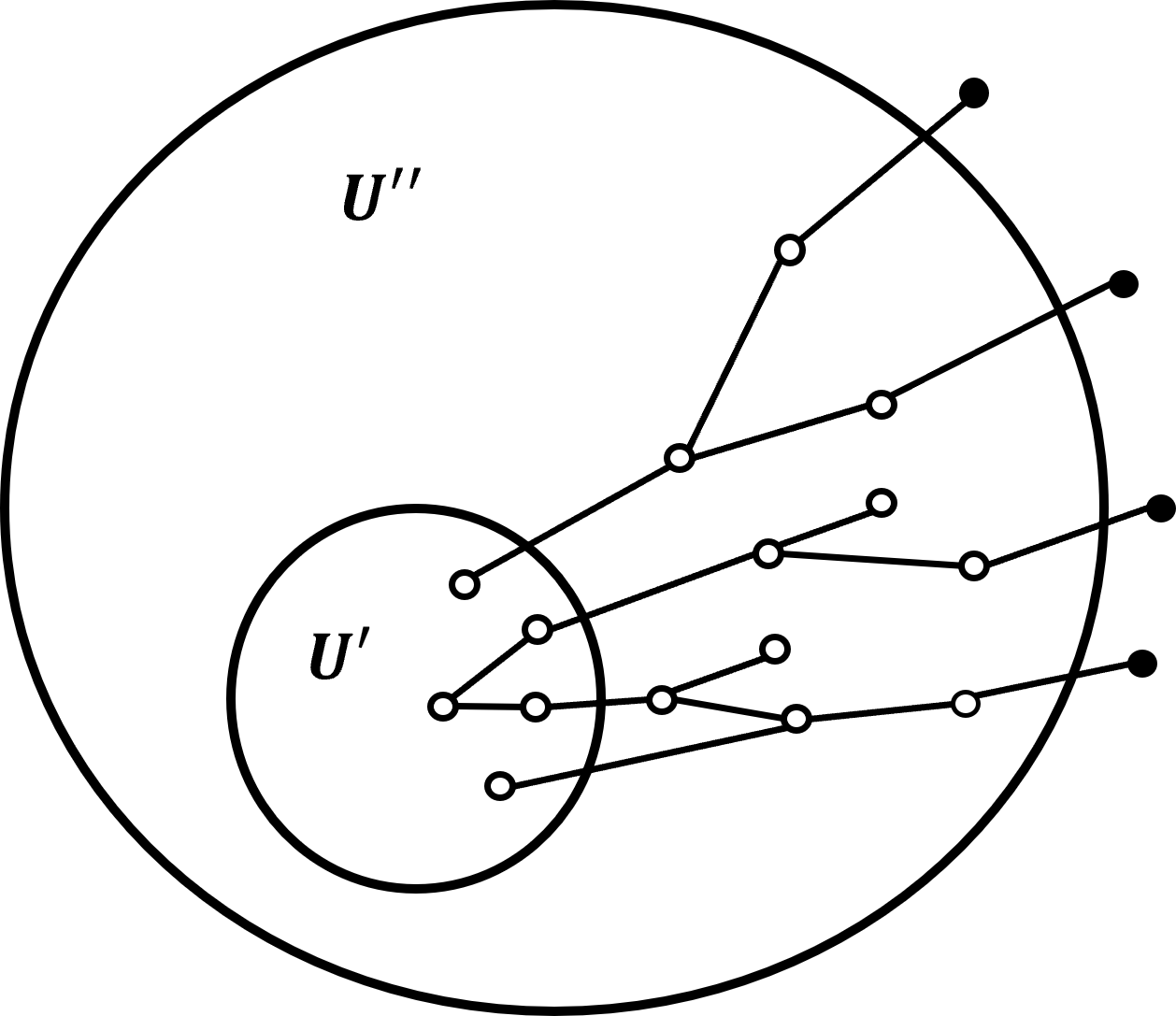}
\caption{$H_c(U)$ with $R=3$. We want to disconnect $U'$ from all nodes with distance $R$ from $U'$ (illustrated as black nodes).}\label{fig:cut_2}
\end{figure}

We will show the following main result for Algorithm~\ref{alg:main}:

\begin{theorem}
\label{combine-thm0}
If every edge has a palette of size $ (1+\eps) \arb$, then w.h.p. Algorithm~\ref{alg:main} generates a list-forest decomposition of $G^{\text{main}}$ such that $\arb^*(G^{\text{leftover}}) \leq  \eps \arb $. It has the following complexity:
\begin{itemize}
\item If $\eps \arb \geq \Omega(\log n)$, the complexity is $O(\frac{\log^4 n}{\eps})$ rounds.
\item If $\eps \arb \geq \Omega(\log n)$ and $\cs = \{1, \dots,  \arb (1 + \eps)  \}$, the complexity is $O(\frac{\log^3 n}{\eps})$ rounds.
\item If $\eps \arb \leq \log \Delta$, the complexity is $O(\frac{\Delta^{2/(\eps \arb)} \log \Delta \log^4 n}{\arb \eps^2})$ rounds.
\item If $\eps \arb \geq \log \Delta$,  the complexity is $O(\frac{\log^4 n}{\eps})$ rounds.
\end{itemize}
\end{theorem}

The key to the algorithm is to ensure that the  $\cut$ subroutine load-balances the number of removed out-neighbors of any vertex.   We describe implementation of $\cut$ to achieve this, along with choices of parameter $R$, in Section~\ref{sec32}.  At the end of this process, we finish with a decomposition of the leftover graph; this is summarized next in  Section~\ref{combine-sec}.   Putting aside the implementation of $\cut$ for the moment, we summarize the algorithm as follows:
\begin{theorem}
\label{alg2-thm}
Algorithm~\ref{alg:main} runs in  $O(R \log^2 n +  \log^3 n / \eps)$ rounds. If the execution of the algorithm is good and every edge has a palette of size $(1+\eps) \arb$, then at the termination, $\psi$ is a list forest decomposition of $G^{\text{main}}$.
\end{theorem}

\begin{proof}
Let $\bar R = R + R'$. For Line \ref{step:nd}, we use the algorithm of \cite{EN16} to get an $(O(\log n), O(\log n))$-network decomposition in the power graph $G^{2 \bar R}$ in $O( \bar R \log^2 n)$ rounds. Then Algorithm~\ref{alg:main} colors all edges that are adjacent to or inside a cluster $U$ of a class $z$ (Line \ref{ln:outerloopstart} to Line \ref{ln:lastline}).  Thus, if an edge $uv$ is not removed, it will become colored when we process the first class containing $u$ or $v$.

 Consider some cluster $U$, and suppose the execution is good. The modified coloring $\psi'$ can be found by some ``leader'' vertex in $U$, which reads in the neighborhood $U'' = N^{\bar R}(U)$. By Corollary~\ref{aug:cor}, it is possible to modify edge colors within $U'$ so that all edges in $U$ become colored, for large enough constant $K'$.  Note that, since there are no paths in $H_c(U)$ from $U$ to outside $U''$, we can check whether the coloring is acyclic by looking within $U''$ alone.

The distance between clusters in the same class is at least $2 \bar R + 1$. Moreover,  if $u,v$ are adjacent in $G^{2 \bar R}$, their distance in $G$ is at most $2 \bar R$. So each cluster $U$ has weak diameter at most $O(\bar R \log n)$, and also the balls $N^{\bar R}(U_1)$ and $N^{\bar R}(U_2)$ must be disjoint for any two clusters  $U_1$ and $U_2$ of the same class. We can process each cluster independently without interfering with others.   Therefore, Line \ref{ln:outerloopstart} to Line \ref{ln:lastline}, including implementation of $\cut$, can be simulated locally in $ O( \bar R \log n)$ rounds.  Since there are $O(\log n)$ classes, the total running time is $O(\bar R \log^2 n)$.
\end{proof}

\subsection{Implementing $\cut$} 
\label{sec32}
 Let us define $T = O( \log n )$ to be the number of classes in the network decomposition.  We now describe a few strategies to implement $\cut$, with different parameter choices for the radius $R$. We summarize these rules as follows:

\begin{theorem}
\label{thm32}
The procedure $\cut$ can be implemented so  that w.h.p. the leftover subgraph has pseudo-arboricity at most $ \eps \arb $ and the execution of Algorithm~\ref{alg:main} is good,  with the following values for parameter $R$:
\begin{enumerate}
\item $R = O ( \frac{\log^2 n}{\eps})$ if $\eps \arb\geq \Omega(\log n)$. 
\item $R = O( \frac{\log n}{\eps})$ if $\eps \arb \geq \Omega(\log n)$ and $\cs = \{1, \dots,  \arb (1 + \eps)  \}$ (i.e. forest decomposition).
\item $R = O(\frac{\Delta^{2/(\eps \arb)} \log \Delta \log^2 n}{\arb \eps^2})$ if $\eps \arb \leq \log \Delta$
\item $R = O(\frac{\log^2 n}{\eps})$ if $\eps \arb \geq  \log \Delta$
\end{enumerate}
\end{theorem}

Theorem~\ref{combine-thm0} will follow directly from Theorem~\ref{alg2-thm} and Theorem~\ref{thm32}. We show Theorem~\ref{thm32} here; the first two results follow from straightforward diameter-reduction algorithms.
\begin{proof}[\textnormal{\textbf{Proof of Theorem~\ref{thm32}(1)}}]
We apply Proposition~\ref{decompose-prop3} to  $H(U)$ with parameter $\eps' =  \eps/(2 T)$ in place of $\eps$. This reduces the diameter of each forest $H_c(U)$ to $D = O( \log^2 n /  \eps )$ and removes an edge-set of arboricity at most $\lceil \eps \arb / (2 T) \rceil$.  In particular, when $R \geq D$, the execution of Algorithm~\ref{alg:main} is good. Over the $T$ iterations of Algorithm~\ref{alg:main}, the arboricity of $G^{\text{leftover}}$ is at most $T \cdot \lceil \eps \arb / (2 T) \rceil$; since $T = O(\log n)$ and $\eps \arb \geq \Omega(\log n)$, this is at most $\eps \arb$.
\end{proof}

\begin{proof}[\textnormal{\textbf{Proof of Theorem~\ref{thm32}(2)}}]
For each color $c$, we choose an arbitrary root for each tree of $H_c(U)$. Next, we choose an integer $J_c$ uniformly at random from $[N]$, where $N = \lceil 4 T / \eps \rceil$, and set $R = 2 N + 1= \Theta(\frac{\log n}{\eps})$. Then $\cut(U, R)$ removes all edges $e$ in each $H_c(U)$ whose tree-depth $d_e$ satisfies $d_e \equiv J_c \mod  N $. After this deletion step, each $H_c(U)$ has path length at most $2 N < R$. So $V \setminus U''$ is disconnected from $U'$ with probability one and Algorithm~\ref{alg:main} is always good.

When the algorithm removes any edge $e = uv$, where $u$ is the parent of $v$ in the rooted tree of $H_c(U)$, we can orient edge $e$ away from $v$ in $G^{\text{leftover}}$. The outdegree of $v$ in $G^{\text{leftover}}$ is then $\sum_{i=1}^{T} \sum_{c \in \cs} Y_{i,c}$, where $Y_{i,c}$ is the indicator function that $v$ has its $c$-colored parent edge removed when processing class $i$.  For a subset $S$ of $[T] \times \cs$, we have $\Pr(\bigwedge_{(i,c) \in S} Y_{i,c} = 1)\leq q^{|S|}$ where $q = 1/N$.  Note that $|\cs| T q \leq \eps \arb / 2$, so by Lemma~\ref{chernoff-ext}, the probability that the outdegree of $v$ exceeds $\eps \arb$ is at most $F_+( \eps \arb/2, \eps \arb) \leq e^{-\eps \arb/6}$. When $\eps \arb \geq \Omega(\log n)$, then w.h.p. every vertex has at most $\eps \arb$ out-neighbors in the orientation.
\end{proof}

 We now turn to the last two results of Theorem~\ref{thm32}.   We assume that $\eps \arb = O(\log n)$, as otherwise we could apply Theorem~\ref{thm32}(1). In particular, from our assumption that $\eps > 1/n$, we have $\arb \leq O(n \log n)$. The algorithm for $\cut(U, R)$ here has two stages: an initialization procedure, which is called at the beginning of Algorithm~\ref{alg:main}, and an on-line procedure for a given cluster $U'$. 
 
 \begin{algorithm}[H]
\caption{\textsc{CUT$(U, R)$}}\label{alg:cut-im}
\begin{algorithmic}[1]\small
\Function{Initialization}{}
\State Using Theorem~\ref{ace1prop}, obtain a $3 \arb$-orientation $J$ of $G$

\State For each vertex $v \in V$, set $L(v) \leftarrow 0$
\EndFunction
\Function{\textsc{CUT$(U, R)$ execution}}{}
\For{each vertex $v \in U''$ with $L(v) < \eps \arb$}

\State Draw independent Bernoulli-$p$ random variable $X_v$
\If{$X_v = 1$} 

\State Select an edge $e$ uniformly at random from the out-going edges of $v$ under $J$.

\State Remove edge $e$ from $G$

\State $L(v) \leftarrow L(v) + 1$
\EndIf
\EndFor 
\EndFunction
\end{algorithmic}
\end{algorithm}

We say a vertex $u$ is {\it overloaded} if $L(u) \geq \eps \arb$, otherwise it is \emph{underloaded}; thus, Algorithm~\ref{alg:cut-im} only modifies underloaded vertices.  For an edge $e$ oriented from $u$ to $v$ in $J$, we say that $e$ is {\it overloaded} or {\it underloaded} if $u$ is.  Given a path $P$, we let $E_0(P)$ and $E_1(P)$ denote the set of underloaded and overloaded edges in $P$ respectively.  A length-$R$ path in $H_c(U)$ is called a {\it live branch}. 
\begin{proposition}
\label{succ-lemma}
Let $\eta \in (0,1/2]$.  If $p \geq \frac{30 \arb \log n}{\eta R}$, then w.h.p., either the execution of Algorithm~\ref{alg:main} is good, or some live branch $P$ has $|E_0(P)| < \eta R$.
\end{proposition}
\begin{proof}
Any path from $U'$ to $V \setminus U''$ has length at least $R$, hence will pass over some live branch. So it suffices to show that any live branch $P$ in $H_c(U)$ during an invocation of $\cut(U,R)$ is cut. Each underloaded edge of $P$ gets removed with probability at least $p/(3 \arb)$, and such removal events are negatively correlated. Thus, for $|E_0(P)| \geq \eta R$, the probability that $P$ remains is at most $(1 - p/(3 \arb))^{\eta R} \leq e^{-p R \eta /(3 \arb)}$. By our choice of $p$, this is at most $e^{-30/3 \log n} \leq n^{-10}$.

Each forest $H_c(U)$ has at most $n^2$ live branches, and Algorithm~\ref{alg:main} invokes $\cut(U,R)$  at most $O(n \log n)$ times, and the number of colors $c$ is at most $m  \arb (1+\eps) \leq 2 n \arb^2 \leq O(n^3 \log^2 n)$. Hence, by a union bound, we conclude the algorithm is good or some live branch has $|E_0(P)|<  \eta R$.
\end{proof}

\begin{lemma}
\label{lem37}
If $R \geq \Omega( \frac{ \Delta^{\frac{2 + 4 \eta}{ \eps \arb }}  \log^2 n}{\eta \eps})$ for some $\eta \in (0,1/2]$, then $p$ can be chosen so that Algorithm~\ref{alg:main} is good w.h.p. 
 \end{lemma}
\begin{proof}
Let $t =  \eps \arb$ and $p = \frac{30 \arb \log n}{\eta R}$. We set $R = \frac{K \Delta^{\frac{2 + 4 \eta}{t }}  \log^2 n}{\eta \eps}$ for some constant $K$, and we can calculate:
\begin{equation}
\label{peqn}
p  = \frac{30 \arb \log n}{\eta R} \leq \eps \arb \cdot \frac{30}{K} \cdot \frac{1}{\log n \cdot \Delta^{\frac{2 + 4 \eta}{t}}} ~.
\end{equation}

Since we are assuming $\eps \arb \leq O(\log n)$, we have $p \in [0,1]$ for large enough $K$. By \Cref{succ-lemma}, it now suffices to show that w.h.p. $|E_1(P)| < (1-\eta) R$ for all live branches $P$. 

Consider the probability that all edges in $S$ are overloaded where $S$ is an arbitrary subset of the edges in a given live branch $P$.   Since $P$ is a path, $S$ involves at least $|S|/2$ distinct vertices. For each such vertex $u$, the value $L(u)$ is a truncated Binomial random variable with mean at most $T p$. Hence $u$ is overloaded with probability at most $r = F_+(T p, t)$.  Accordingly, the probability that all edges in $S$ are overloaded is at most $r^{|S|/2}$.

Since $T \leq O(\log n)$, Eq.~(\ref{peqn}) implies that $p T \leq \frac{\eps \arb}{10 e \Delta^{\frac{2 +4 \eta}{t}}}$ for large enough $K$, and therefore
$$
F_+( T p, t) \leq \Bigl( \frac{e T p}{t} \Bigr)^t \leq \Bigl( \frac{e  \cdot \eps \arb}{10 e \Delta^{(2+4 \eta)/t} \cdot  \eps \arb } \Bigr)^t \leq \frac{1}{10 \Delta^{2 + 4 \eta}}~.
$$

So we apply Lemma~\ref{chernoff-ext} with parameter $q = \sqrt{r}$ for the random variable $|E_1(P)|$ to get:
\begin{align*}
&\Prob{ |E_1(P)| > (1-\eta) R} \leq F_+(R \sqrt{r}, (1-\eta) R) \leq \Bigl( \frac{e  \sqrt{r}}{1-\eta} \Bigr)^{(1 - \eta) R}  \\
&\qquad \qquad \leq  \Bigl(  \frac{e}{\sqrt{10} (1-\eta) \Delta^{1 + 2 \eta}} \Bigr)^{(1-\eta)R} \leq (e/\sqrt{10})^{R/2} \Delta^{-(1-\eta) (1+2 \eta) R} \leq (0.93 / \Delta)^R  ~.
 \end{align*}
 
 Since $\eta \leq 1/2$ and $R \geq \omega(\log n)$, this is at most $ \poly(1/n) \cdot \Delta^{-R}$.   There are at most $n \Delta^{R-1}$ paths of length $R$. By a union bound, we conclude that w.h.p. $|E_0(P)| \geq \eta R$  for all such paths.   
\end{proof}

\begin{proof}[\textnormal{\textbf{Proof of Theorem~\ref{thm32}(3),(4)}}] 
In the algorithm for $\cut$, each vertex removes at most $ \eps \arb $ outgoing neighbors under $J$. Hence, $\arb^*(G^{\text{leftover}}) \leq  \eps \arb  $. Given $\eta$, we choose $R, p$ according to Lemma~\ref{lem37} so that Algorithm~\ref{alg:main} is good w.h.p.  For result (3), we set $\eta = \frac{  \eps \arb }{2 \log \Delta}$, giving $R  \leq O( \frac{\Delta^{2/(\eps \arb)} \log \Delta \log^2 n}{\arb \eps^2})$.   For result (4),  we set $\eta = 1/2$, giving $R \leq O(\frac{\log^2 n}{\eps})$.
\end{proof}

We remark that, by orienting edges  in terms of the forest $H_c(U)$ instead of the fixed orientation $J$, the bound on $R$ can be reduced to $\frac{\Delta^{1/(\eps \arb)} \log \Delta \log^2 n}{\arb \eps^2}$; this leads to an $O( \frac{\Delta \log \Delta \log^4 n}{\eps} )$-round algorithm for ordinary forest-decomposition when $\eps \arb \geq 3$. We omit this analysis here.

\subsection{Putting Everything Together}
\label{combine-sec}

We now need to combine the forest decomposition of the main graph with a forest decomposition on the leftover graph.  For ordinary coloring, this is straightforward; we summarize it as follows:
\begin{theorem}
\label{combine-thm1}
We can obtain an $(1+\eps) \arb$-FD of $G$ of diameter $D$, under the following conditions:
\begin{itemize}
\item If $\eps \arb \geq 3$, then $D \leq n$, and the complexity is $O( \frac{\Delta^2 \log \Delta \log^4 n}{\eps} )$.
\item  If $4 \leq \eps \arb \leq \log \Delta$, then $D \leq O( \frac{\log n}{\eps})$, and the complexity is $\frac{\Delta^{O(1/\eps \arb)} \log \Delta \log^4 n}{\arb \eps^2}$.
\item If $\eps \arb \geq \log \Delta$,  then $D \leq O( \frac{\log n}{\eps})$, and the complexity is $O(\frac{\log^4 n}{\eps})$.
\item If $\eps^2 \arb \geq \Omega(\log \Delta)$,  then $D \leq O( \frac{1}{\eps})$, and the complexity is $O(\frac{\log^4 n}{\eps})$.
\item If $\eps \arb \geq \Omega(\log n)$, then $D \leq O( \frac{1}{\eps})$, and the complexity is $O(\frac{\log^3 n}{\eps})$.
\end{itemize}
\end{theorem}
\begin{proof}
The first step for all these results is to apply Theorem~\ref{combine-thm0} where each edge is given the palette $ \{1, \dots, \arb + \lceil \eps \arb / 10 \rceil \}$. Then $G^{\text{leftover}}$ has pseudo-arboricity at most $k' = \lceil \eps \arb / 10 \rceil$, and  Theorem~\ref{ace1prop}(3) yields a $\lfloor 2.01 k' \rfloor$-FD of these leftover edges.  Taken together, these give a $k$-FD of $G$ for $k = \arb + \lfloor 2.01 \lceil \eps \arb / 10 \rceil \rfloor + \lceil \eps \arb / 10 \rceil$. 

 The runtime bounds follow immediately from the four different cases of Theorem~\ref{combine-thm0} .
 
This immediately gives us the first result in our list.  For the next four results, we apply Proposition~\ref{decompose-prop3} to convert this into a $k + \lceil \eps \arb / 10 \rceil$-FD of $G$, with the given bounds on the diameter.
\end{proof}
 
 For list-coloring, we will combine the main graph and leftover graph by partitioning the color-space $\cs$. Specifically, each vertex $v$ will choose a color set $C_v \subseteq \cs$; we also write $\vec C$ for the ensemble of values $C_u: u \in V$. Given $\vec C$, we define new color palettes for each edge $e = uv$ by
 \begin{align*}
 Q^{\text{main}}(uv) &= Q(uv) \cap C_u \cap C_v \\
Q^{\text{leftover}} (uv) &= Q(uv) \setminus (C_u \cup C_v) 
 \end{align*}
 
We can now describe  two main algorithms for this type of color partition.
\begin{theorem}
\label{gen-partition0}
Suppose that each edge has a palette of size $(1+\eps) \arb$.  Then, w.h.p., we can choose $\vec C$ in $O(\frac{\log n}{\eps})$ rounds with the following palette sizes:
\begin{itemize}
\item If $\eps \arb \geq \Omega(\log n)$, every edge $e$ has $|Q^{\text{main}}(e)| \geq  (1+\eps/2) \arb$ and $|Q^{\text{leftover}}(e)| \geq \eps \arb / 20$.

\item If $\eps^2 \arb \geq \Omega(\log \Delta)$, every edge $e$ has $|Q^{\text{main}}(e)|  \geq (1+\eps/2) \arb$ and $|Q^{\text{leftover}}(e)| \geq \eps^2 \arb / 200$.
\end{itemize}
\end{theorem}
\begin{proof}
For the first result, we independently draw a $( O(\frac{\log n}{\eps}), \frac{\eps}{10} )$-stochastic network decomposition $L_c$ for each color $c$, and for each connected component $U$ in the graph $(V,  L_c)$ we draw a Bernoulli-$(\eps/10)$ random variable $X_{c,U}$. Then each vertex $v \in U$ has $c \in C_v$ if and only if $X_{c,U} = 0$.

Now consider some edge $e=uv$ and color $c \in Q(e)$. If $e \in L_c$, then vertices $u$ and $v$ are in the same component $U$ of $(V,  L_c)$, so $c$ is in $Q^{\text{main}}(e)$ or $Q^{\text{leftover}}(e)$ depending on the value $X_{c,U}$.  Thus, $c$ goes into $Q^{\text{main}}(e), Q^{\text{leftover}}(e)$ with probability at least $(1 - \eps/10)  (1-\eps/10) \geq 1 - \eps/5$ and $(1-\eps/10)  \eps/10 \geq \eps/11$, respectively. Since each color operates independently, Chernoff bounds imply that $|Q^{\text{main}}(e)|$ and $|Q^{\text{leftover}}(e)|$ have respective sizes at least $\arb(1 + \eps/2)$ and $\eps \arb/20$ with probability at least $1 - e^{-\Omega(\eps \arb)}$. There are $m \leq n \arb$ edges, so the desired bounds hold with probability at least $1 -  n \arb e^{-\Omega(\eps \arb)}$; since $\eps > 1/n$, this is $1 - 1/\poly(n)$ for $\eps \arb \geq \Omega(\log n)$.

For the second result, we draw an independent Bernoulli-$(\eps/10)$ random variable $X_{c,v}$ for each color $c$ and vertex $v$, and we place $c$ in $C_v$ if and only if $X_{c,v} = 0$. For any edge $e$, the expected size of $Q^{\text{main}}(uv)$ is $\arb(1+\eps)(1 - \eps/10)(1-\eps/10) \geq \arb (1 + \eps/5)$ and the expected size of $Q^{\text{leftover}}(e)$ is  $\arb (1+\eps) (\eps/10)  (\eps/10) \geq \eps^2 \arb/100$. We can use the LLL algorithm of \cite{chung2017distributed}, where each edge $e$ has a bad-event $B_e$ that  $|Q^{\text{main}}(e)| < (1 + \eps/2) \arb$ or $|Q^{\text{leftover}}(e)| < \eps^2 \arb / 200$. When $\eps^2 \arb \geq \Omega(\log \Delta)$, a straightforward Chernoff bound shows  $p \leq \Delta^{-11}$. Also, $B_e$ affects at most $d = 2 \Delta$ other bad-events (corresponding to its neighboring edges). So  the criterion $p d^2 \ll 1$ holds and the LLL algorithm runs in $O(\log n)$ rounds.
\end{proof}

These give the following final results for LFD:
\begin{theorem}
\label{combine-thm3}
Suppose that $G$ is a multigraph where each edge has a palette of size $(1+\eps) \arb$. We can obtain an LFD of $G$ of diameter $D$ w.h.p., under the following conditions:
\begin{itemize}
\item If $\eps \arb \geq \Omega(\log n)$, the complexity is $O( \frac{\log^4 n}{\eps} )$ rounds and $D = O( \frac{\log n}{\eps})$.
\item If $\eps^2 \arb \geq \Omega(\log \Delta)$, the complexity is $O( \frac{\log^4 n}{\eps^2} )$ rounds and $D = O( \frac{\log n}{\eps^2})$.
\end{itemize}
\end{theorem}
\begin{proof}
For the first result, let $\eps' = \eps/1000$. We begin by applying Theorem~\ref{gen-partition0}, obtaining palettes $Q^{\text{main}}, Q^{\text{leftover}}$. We then apply Theorem~\ref{combine-thm0} with respect to palettes $Q^{\text{main}}$; given our bound $|Q^{\text{main}}(e)| \geq (1 + \eps/2) \arb \geq (1 + \eps') \arb$ for all $e$ and $\eps' \arb \geq \Omega(\log n)$, this produces an LFD $\phi$ of $G^{\text{main}}$, along with a leftover graph with pseudo-arboricity at most $\eps' \arb $.  

Next, we apply Proposition~\ref{decompose-prop3a1} to $\phi$ with parameter $\eps'$ to obtain an edge-set $E' \subseteq E^{\text{main}}$ such that $\arb(E') \leq \lceil \eps' \arb \rceil$ and $\phi$ has diameter $O( \frac{\log n}{\eps})$ on $E^{\text{main}} \setminus E'$.  Finally, we apply Theorem~\ref{lem:Hpartition3} to edge-set $E'' = E^{\text{leftover}} \cup E'$ to get an LSFD $\phi''$ of $E''$ with palettes $Q^{\text{leftover}}$; note that $\arb^*(E'') \leq \arb^*(E^{\text{leftover}}) + \arb^*(E') \leq 2 \lceil \eps' \arb \rceil$ and $|Q^{\text{leftover}}(e)| \geq \eps \arb/20 \geq 5 \arb^*( E'' )$ for all $e \in E''$.

Now consider the coloring $\bar \phi$ defined by setting $$
\bar \phi(e) = 
\begin{cases}  \phi''(e) & \text{if $e \in E''$} \\
\phi(e) & \text{if $e \in E \setminus E''$}
\end{cases}
$$
 We claim that any component of any color $c$ can only contain edges from $E''$ or $E \setminus E''$, but not both. For, suppose some vertex $v$ has $c$-colored edges $vu, vu''$ in $E \setminus E'', E''$ respectively. Since $\bar \phi(vu) = \phi(vu)$, we have $c \in Q^{\text{main}}(vu) \subseteq C_v$, but since $\bar \phi(vu'') = \phi''(vu'')$, we have $c \in Q^{\text{leftover}}(v u'') \subseteq \mathfrak C \setminus C_v$. This is a contradiction.
 
In particular, $\bar \phi$ is an LFD of the full graph $G$ and  the diameter of $\bar \phi$ is the maximum of the diameters of $\phi$ on $E \setminus E''$ and $\phi''$ on $E''$. 

The second result is completely analogous, except we set $\eps' = \eps^2/1000$. 
\end{proof}

\section{Star-Forest Decomposition for Simple Graphs}
\label{star-forest-sec}
Let $G$ be a simple graph of arboricity $\arb$. By using Theorem~\ref{orient-cor}, we may assume that we have obtained some $t$-orientation $J$ in $O(\log^3 n/\eps)$ rounds, where $t = \lceil (1 + \eps) \arb \rceil$. We write $J(v)$ for the set of out-neighbors of each vertex $v$; by adding dummy directed edges as necessary, we may assume that $|J(v)| = t$ exactly. 

To obtain a star-forest-decomposition of the graph, consider the following process:  each vertex $v$ in the graph selects a color set $C_v \subseteq \cs$; again, we write $\vec C$ for the ensemble of values $C_u: u \in V$. 
For each vertex $v$, we construct a corresponding bipartite graph $W_v(\vec C)$, whose left-nodes correspond to $\cs$ and whose right-nodes correspond to $J(v)$, with an edge from left-node $c$ to right-node $u$ if and only if $c \in Q(uv) \cap (C_v \setminus C_u)$. We have the key observation:

\begin{proposition}
\label{thm1a}
Let $\delta \in \mathbb Z_{\geq 0}$ be a globally-known parameter. If each $W_v( \vec C)$ has a matching of size at least $t-\delta$, then in $O(1)$ rounds we can generate a partition of the edges $E = E_0 \sqcup E_1$ along with a LSFD $\phi_0$ of $E_0$,  such that $\arb^*(E_1) \leq \delta$.  (In particular, if $\delta = 0$, then $\phi_0$ is an LSFD of $G$).
\end{proposition}
\begin{proof}
For each edge $(c,u)$ in the matching $M_v$ of $W_v( \vec C)$, we set $\phi(vu) = c$. Thus, all color-$c$ edges have the form $vu$ where $c \in C_v \setminus C_u$ and $(c,u) \in M_v$.  For $c \in C_v$, we say that $v$ is a $c$-leaf and for $c \notin C_v$ we say that $v$ is a $c$-center.  Since $M_v$ is a matching, the edges of each color $c$ are a collection of stars on the $c$-centers and $c$-leaves. The original $t$-orientation $J$ of $G$ in turn yields a $\delta$-orientation for the residual uncolored edges.
\end{proof}

So we need to select $\vec C$ so that every graph $W_v( \vec C)$ has a large matching.  The following two results show how to achieve this by random sampling; the precise details are different for list and ordinary coloring. Given a fixed choice for $\vec C$, we write $W_v$ as shorthand for $W_v(\vec C)$.
\begin{lemma}
\label{sclem1}
Suppose that $\eps \arb \geq 100 (\sqrt{\log \Delta} + \log \arb)$. If $\cs = [t]$ and each set $C_u$ is chosen uniformly at random among $\arb$-element subsets of $\cs$, then for any vertex $v$ there is a probability of at least $1 - 1/\Delta^{10}$ that $W_v( \vec C)$ has a matching of size at least $\arb (1 - \eps)$.
\end{lemma}
\begin{proof}
Let us suppose that we have fixed $C_v$ to some arbitrary value.  By a slight extension of Hall's theorem, it suffices to show that any set $X \subseteq J(v)$ has at least $|X| - 2 \eps \arb$ neighbors in $W_v$; equivalently, there is no pair of sets $X \subseteq J(v), Y \subseteq C_v$ with $|X| \geq 2 \eps \arb$ and $|Y| = \arb (1 + 2 \eps) - |X|$ such that $W_v$ contains no edges between $X$ and $Y$. We say that such a pair $X,Y$ is \emph{bad}. For any fixed $X, Y$ with $x = |X|, y = |Y| = \arb(1+2 \eps) - x$, the probability that $X, Y$ is bad is given by
$$
\Bigl( \tbinom{t - y}{\arb - y} / \tbinom{t}{\arb} \Bigr)^x =  \Bigl( \tbinom{t -y }{\eps \arb} / \tbinom{t}{\eps \arb} \Bigr)^x \leq (1 - y/t)^{x \eps \arb} \leq e^{-x y \eps \arb/t} \leq e^{-x y \eps/2}.
$$

Summing over the $\binom{\arb (1+\eps)}{x}$ choices for $X$  and $\binom{\arb}{y}$ choices for $Y$ of a given cardinality, the total probability of any bad pair $X, Y$ is at most:
\begin{equation}
\label{xyeq0}
\sum_{x = 2 \eps \arb}^{\arb(1+\eps)}  \binom{\arb (1+\eps)}{x} \binom{\arb}{\arb(1+2 \eps) - x} e^{-x (\arb(1 + 2 \eps) - x) \eps/2}~.
\end{equation}

When  $2 a \eps \leq x \leq \arb/2$, the summand in Eq.~(\ref{xyeq0}) is at most
\begin{align*}
&\tbinom{\arb (1+\eps)}{x} \tbinom{\arb}{\arb(1 + 2 \eps) - x} e^{-x (\arb (1 + 2 \eps) - \arb/2) \eps/2} \leq \tbinom{\arb (1+\eps)}{x} \tbinom{\arb}{x - 2 \eps \arb}  e^{-x \eps \arb/4}  \leq (\arb(1+\eps))^x \arb^x e^{-x \eps \arb/4} \\
&\qquad \qquad \leq (2  \arb^2  e^{-\eps \arb/4} )^x \leq (2  \arb^2  e^{-25 \log \arb - 25 \sqrt{\log \Delta}} )^x \leq (2 \arb^{-23} e^{-25 \sqrt{\log \Delta}})^x \\
& \qquad \qquad  \leq (e^{-25 \sqrt{\log \Delta}})^x  \leq (e^{-25 \sqrt{\log \Delta}})^{2 \eps \arb}  \leq (e^{-25 \sqrt{\log \Delta}})^{200 \sqrt{\log \Delta}} = \Delta^{-5000}.
\end{align*}

In a completely analogous way (but with slightly different numerical constants), the summand of Eq.~(\ref{xyeq0}) for  $\arb/2 \leq x \leq \arb (1 + \eps)$ is at most $\Delta^{-2500}$. Since there are at most $\Delta$ such summands, the overall sum is at most $\Delta^{-2499}$.
\end{proof}

\begin{lemma}
\label{sclem2}
Let $\eps \leq 10^{-6}, \Delta \geq 10^6$, and $\eps \arb \geq 10^6 \log \Delta$.  Suppose  each edge has a palette of size $\arb(1 + 200 \eps)$. If we form each set $C_u$ by selecting each color independently with probability $1 - \eps$, then for any vertex $v$ there is a probability of at least $1 - 1/\Delta^{10}$ that $W_v(\vec C)$ has a matching of size $t$. 
\end{lemma}
\begin{proof}
Let $s = \arb(1 + 100 \eps)$. We say that $C_v$ is \emph{good} if $|Q(uv) \cap C_v| \geq s$ for all $u \in J(v)$. We first claim that $C_v$ is good with probability at least $1 - \Delta^{-100}$. For, observe that $|Q(uv) \setminus C_v|$ is a Binomial random variable with mean $|Q(uv)| \eps \leq 2 \eps \arb$. Hence, $\Pr(|Q(uv) \setminus C_v| \geq 100 \eps \arb) \leq F_+( 2 \eps \arb, 100 \eps \arb) \leq ( \frac{e \cdot 2 \eps \arb}{100 \eps \arb})^{100 \eps \arb} \leq 2^{-100 \eps \arb}$; by our assumption $\eps \arb \geq 10^6 \log \Delta$, this is at most $\Delta^{-10^8}$. By a union bound over $u \in J(v)$, the probability that $C_v$ is good is at least $1 - t \cdot \Delta^{-10^8} \geq 1 - \Delta^{-100}$.

We next argue that, conditioned on a fixed choice of good $C_v$, there is a probability of at least $1 - \Delta^{-11}$ that $W_v$ has a $t$-matching. By Hall's theorem, it suffices to show that any set $X \subseteq J(v)$ has at least $|X|$ neighbors in $W_v$; define $B_X$ to be the bad-event that this condition fails for $X$. By considering the exponential generating function for the number of neighbors of $X$, similar to a Chernoff bound argument, we can show that any such set $X$ has
$$
\Pr(B_X) \leq  \Bigl(  \frac{s e^{-\eps x}}{s - x } \Bigr)^{s - x } \Bigl( \frac{s (1 - e^{-\eps x})}{x} \Bigr)^x \qquad \qquad \text{where $x = |X|$}
$$
(The details are somewhat involved, so we defer the proof of this fact to Proposition~\ref{nnb4} in the Appendix.)  We can take a union bound over the $\binom{t}{x}$ possible sets $X$ of cardinality $x$, to get:
$$
\sum_{X} \Pr(B_X) \leq \sum_{x=1}^t \beta_x \qquad \text{for } \beta_x :=  \Bigl(  \frac{s e^{-\eps x}}{s - x } \Bigr)^{s - x } \Bigl( \frac{s (1 - e^{-\eps x})}{x} \Bigr)^x \binom{t}{x}~.
$$

We will upper-bound $\beta_x$ separately over two different parameter regimes:

\paragraph{Case I: $\pmb{\arb/3 \leq x \leq t}$} Here, we have $s e^{-\eps x} \leq (2 \arb) e^{\eps \arb /3 } \leq (2 \Delta) e^{-10^6 (\log \Delta) / 3} \leq 2\Delta^{-1000}$ and $\binom{t}{x} \leq t^{t-x} \leq (2 \Delta)^{t-x}$. Thus, we  use the upper bound:
$$
\beta_x \leq f_1(x) := \bigl( 2 \Delta^{-1000} \bigr)^{s-x} (s/x)^x (2 \Delta)^{t-x},
$$

Letting $g_1(x) = \log f_1(x)$, we compute the derivative $g'_1(x) = -1 + 999 \log \Delta + \log s - \log(4 x)$; since $x \leq s$ and $\Delta \geq 10^6$ this is at least $100 \log \Delta$. Hence, in this range, we have $\beta_x \leq e^{g_2(x)} \leq e^{g_1(s) - 100 (s-x) \log \Delta}$. Since $g_1(s) = (t-s) \log (2 \Delta) \leq 0$, we have $\beta_x \leq \Delta^{-100 (s-x)}$, and the overall sum $\sum_{x=a/3}^{t} \beta_x$ is at most $2 \Delta^{-100 (s-t)} \leq 2 \Delta^{-100}$. 

\paragraph{Case II: $\pmb{1 \leq x \leq \arb/3}$} Here $s(1 - e^{-\eps x}) / x \leq s$ and $\binom{t}{x} \leq t^x$, giving the upper bound:
$$
\beta_x \leq f_2(x) :=  (s e^{-\eps x} / (s-x))^{s-x} (s t)^x,
$$

Again, letting $g_2(x) = \log f_2(x)$, we compute the derivative $g'_2(x) = 1 - \eps s + 2 \eps x + \log t + \log(s-x)$. Now $x \leq \arb/3$ while $\arb \leq s, t \leq 2 \Delta$; along with our bound $\eps \arb \geq 10^6 \log \Delta$, we get $g_2'(x) \leq -100 \log \Delta$ for $x \geq 0$.  Hence, in this range, we have $\beta_x \leq e^{g_2(x)} \leq e^{g_2(0) - 100 x \log \Delta} =  \Delta^{-100 x}$, and so the overall sum $\sum_{x=1}^{a/3} \beta_x$ is at most $2 \Delta^{-100}$.

\bigskip
 
Putting the two cases together, we have shown that, conditional on a fixed good choice of $C_v$, we have $\sum_X \Pr(B_X) \leq 4 \Delta^{-100}$. Since $C_v$ is good with probability at least $1-\Delta^{-100}$, the overall probability that $W_v$ has a $t$-matching is at least $1 - \Delta^{-10}$.
\end{proof}

This leads to our main results for star-forest decomposition:
\begin{theorem}
\label{agab1}
Let $G$ be a simple graph with arboricity $\arb$. 
\begin{itemize}
\item If $\eps \arb \geq \Omega( \sqrt{\log \Delta} + \log \arb )$, then we get an $\arb (1 + \eps)$-SFD in $ O(\frac{\log^3 n}{\eps})$ rounds w.h.p. 
\item If $\eps \arb \geq \Omega(\log \Delta)$, and every edge has palette size $\arb (1+\eps)$, we can get an LSFD in $ O( \frac{\log^3 n}{\eps})$ rounds w.h.p. 
\end{itemize}
\end{theorem}
\begin{proof}
By reparametrizing, it suffices to show that we can get an $\arb (1 + O(\eps))$-SFD of $G$, or an LSFD when every edge has palette size $(1+O(\eps)) \arb$.

After obtaining the orientation $J$, we use the LLL algorithm of \cite{chung2017distributed} to select $\vec C$; here, each vertex $v$ has a bad-event that $W_v(\vec C)$ has maximum matching size less than $t - 2 a \eps$.  By Lemma~\ref{sclem1}, this event has probability at most $p = \Delta^{-10}$ and depends on $d = \Delta^2$ other such events ($u$ and $v$ can only affect each other if they have distance at most $2$). Thus, the criterion $p d^2 \ll 1$ is satisfied and the LLL algorithm runs in $O(\log n)$ rounds.  

Having selected $\vec C$, we apply Proposition~\ref{thm1a} to get a $(1+\eps) \arb$-SFD of $G$, plus a left-over graph of pseudo-arboricity at most $2 \eps \arb$. We finish by applying Theorem~\ref{ace1prop} to get a $6.01 \eps \arb$-SFD of the left-over graph. Overall, we get a $(1 + 7.01 \eps) \arb$-SFD.

The second result is completely analogous, except we use Lemma~\ref{sclem2} to obtain the matchings of $W_v$. In this cases, there is no left-over graph to recolor.
\end{proof}


We remark that the main algorithmic bottleneck for Theorem~\ref{agab1} is obtaining the $t$-orientation. For example, we could alternatively use the algorithm of \cite{SV19b} to obtain the $t$-orientation, and hence obtain the $\arb(1+\eps)$-LSFD, in $\tilde O(\frac{\log^2 n}{ \eps^2})$ rounds.

\section{Lower Bounds on Round Complexity}
In this section, we show lower bounds for the round complexity of randomized \local algorithms for forest decomposition, using the following construction. For given integer parameters $\arb \geq 2$ and $t \geq 1$, we can form a multigraph $G$ beginning with four named vertices $x_1, x_2, y_1, x_2$. We then put $\lfloor \arb/2 \rfloor$ parallel edges from $x_1$ to $x_2$ and $\lfloor \arb/2 \rfloor$ parallel edges from $y_1$ to $y_2$.  We insert a path $P_1$ with $t+2$ vertices from $x_1$ to $y_1$, with $a$ parallel edges between successive vertices on the path (Thus, $P_1$ contains $t$ vertices aside from $x_1, y_1$), and we insert a second path $P_2$ of $t + 2$ vertices arranged in a line from $x_2$ to $y_2$ with $a$ parallel edges. See Figure~\ref{fig4}.

\begin{figure}[H]
\centering
\begin{subfigure}[b]{0.49\textwidth}
\centering
\includegraphics[scale=0.65]{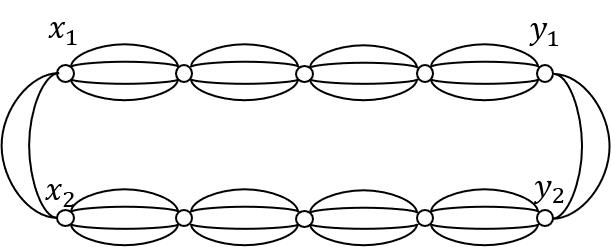}
\caption{}\label{fig:lower_bound_new_g1}
\end{subfigure}
\caption{An illustration of $G$ when $a=4$ and $t=3$.\label{fig4}}
\end{figure}

The graph $G$ has arboricity $\arb$; to see this, consider coloring the edges $x_1$ to $x_2$ by $\{1, \dots \lfloor a/2 \rfloor \}$ and coloring the edges from $y_1$ to $y_2$ by $\{ \lfloor a/2 \rfloor + 1, \dots, 2 \lfloor a/ 2 \rfloor \}$, as well as coloring edges in  $P_1$ and $P_2$ by $\{1, \dots, a \}$. Also, $G$ has $n = 2 t + 4$ vertices and maximum degree $\Delta = O(\arb)$.

\begin{proposition}
\label{missing-prop}
For any $\arb(1+\eps)$-forest-decomposition on $G$, there are at most $2 (t+1) \eps \arb$ colors $c$ where there is a $c$-colored edge between $x_1, x_2$ and also a $c$-colored edge between $y_1, y_2$.
\end{proposition}
\begin{proof}
For any two adjacent nodes $u,v$ in $G$, let us say that color $c$ \emph{appears} if any of the parallel edges between $u$ and $v$ have color $c$, else color $c$ is \emph{missing}. We need to show that at most $2 (t+1) \eps \arb$ colors  appear on both $x_1,x_2$ and $y_1, y_2$.

For consecutive vertices $u,v$ in path $P_1$ or $P_2$, the $\arb$ parallel edges must receive distinct colors (else it would immediately have a cycle). Thus, $u,v$ are missing at most $\eps \arb$ colors. Over the entire paths $P_1, P_2$ which have $2(t+1)$ vertices, there are at most $2(t+1) \eps \arb$ colors missing in total.

But now observe that if a color $c$ appears on all consecutive vertices in the path $P_1, P_2$ as well as between $x_1, x_2$ and between $y_1, y_2$, then the $c$-colored edges would have a cycle. Hence, the only colors which can appear between $x_1, x_2$ and also between $y_1, y_2$ are the ones that are missing from some consecutive vertices in $P_1, P_2$, and there are at most $ 2 (t+1) \eps \arb$ of them.
\end{proof}

Our lower bound will depend in a critical way on using a randomized algorithm which is \emph{oblivious}, i.e. it does not use the provided vertex ID's. In particular, for an $r$-round oblivious randomized algorithm, the output for a given edge $e$ is determined by the isomorphism class of $N^r(e)$. 
\begin{observation}
If any randomized or deterministic \local algorithm can solve a problem in $r$ rounds, then also an oblivious randomized \local algorithm can solve it in $r$ rounds w.h.p.
\end{observation}
\begin{proof}
Given the original algorithm $A$, each vertex chooses a random bit-string of length $K \log n$, and uses it as its new vertex ID for algorithm $A$. W.h.p., all chosen ID's are unique for $K$ sufficiently large, and hence algorithm $A$ succeeds (either with probability one or w.h.p., depending on whether $A$ is randomized).
\end{proof}

\begin{lemma}
\label{double-lemma}
Suppose that $2 (t+1) \eps \arb \leq \lfloor \arb/2 \rfloor$. Then any oblivious algorithm $A$ for $a(1+\eps)$-forest-decomposition on $G$ in less than $t/2$ rounds has success probability at most $\frac{\arb (1+\eps)}{4 (\lfloor \arb/2 \rfloor - 2 (t+1) \eps \arb) }$.
\end{lemma}
\begin{proof}
Let $\ell = \arb(1+\eps)$. For any color $i = 1, \dots, \ell$, let $X_i$ be the indicator function that $(x_1, x_2)$ has an $i$-colored edge and $Y_i$ be the indicator function that $(y_1, y_2)$ has an $i$-colored edge, after we run algorithm $A$ on the graph.   Since the edges $(x_1, x_2)$ and $(y_1, y_2)$ have distance $t$, the random variables $X_i, Y_i$ are independent for each $i$. Furthermore, since the view from $(x_1, x_2)$ is isomorphic to the view from $(y_1, y_2)$ and algorithm $A$ is oblivious,  they follow the same distribution. Thus, we denote $q_i= \E[X_i]  =\E[Y_i]$ and note that $\E[X_i Y_i] = \E[X_i] \E[Y_i] = q_i^2$.

If $A$ returns a forest decomposition, there are $\lfloor \arb/2 \rfloor$ colors between $x_1$ and $x_2$ (a repeated color immediately leads to a cycle), and by \Cref{missing-prop} there are at most $2 (t+1) \eps \arb$ colors  $i$ which appear between $x_1, x_2$ and also between $y_1, y_2$, i.e. which satisfy $X_i = Y_i = 1$. Overall, whenever $A$ returns a forest-decomposition, we have $\sum_{i = 1}^{\ell} X_i = \sum_{i=1}^{\ell} Y_i = \lfloor \arb/2 \rfloor$ and  $\sum_{i=1}^{\ell} X_i Y_i \leq 2 (t+1) \eps \arb$, and in particular  $$
\sum_{i = 1}^{\ell} X_i (1 - Y_i) \geq \lfloor \arb/2 \rfloor - 2 (t+1) \eps \arb.
$$

Let $p$ be the probability that $A$ successfully returns an $\ell$-forest-decomposition. Taking expectations, and noting that $\E[X_i (1 - Y_i)] = \E[X_i] (1 - \E[Y_i]) = q_i (1-q_i)$, we have 
$$
\sum_{i=1}^{\ell} q_i (1-q_i) \geq p ( \lfloor \arb/2 \rfloor - 2 (t+1) \eps \arb ).
$$

On the other hand, clearly $q_i (1-q_i) \leq 1/4$ for all $i$ (since $q_i \in [0,1]$), so $$
\sum_{i = 1}^{\ell} q_i (1 - q_i) \leq \ell / 4.
$$

Putting the bounds together gives $p (\lfloor \arb/2 \rfloor - 2 (t+1) \eps \arb) \leq \ell/4$, which leads to the claimed bound after rearrangement.
\end{proof}

Putting these results together, we obtain the following:
\begin{theorem}\label{thm:lower_bound_multi}
Let $\arb, n \in \mathbb Z_{\geq 2}$ and $\eps \in (0,1)$ with $\eps \arb \geq 1$.  Any randomized algorithm for $(1+\eps)\arb$-forest-decomposition on $n$-node graphs of arboricity $\arb$ with success probability at least $0.51$ requires $\Omega(\min\{ n, 1/\eps \})$ rounds. This bound holds even on graphs of maximum degree $\Delta = O(\arb)$.
\end{theorem}
\begin{proof}
It suffices to show this for an oblivious randomized algorithm $A$ and where $1/n < \eps \leq 0.0001$. In this case, consider forming graph $G$ with parameter $ t = \lceil 0.001/\eps \rceil$; note that $t+2 \leq 0.002/\eps$ due to the upper bound on $\eps$. Thus, $G$ has at most $2 t + 4 \leq 2 (0.002/\eps) \leq n$ nodes. 

Here $\lfloor \arb/2 \rfloor \geq 0.499 \arb$ since $\arb \geq 1/\eps \geq 10000$, and also $2 (t+1) \eps \arb \leq 2 (0.002/\eps) \eps \arb \leq \lfloor \arb/2 \rfloor$.  If $A$ runs in fewer than $t/2$ rounds, then by Lemma~\ref{double-lemma} it has success probability at most $\frac{\arb (1+\eps)}{4 (\lfloor \arb/2 \rfloor - 2 (t+1) \eps \arb) } \leq \frac{\arb (1+\eps)}{4 ( 0.499 \arb - 2 (0.002/\eps) \eps \arb) } \leq 0.506 (1+\eps) < 0.51$.
\end{proof}

We also show that it is impossible to compute $\arb$-FD in $o(n)$ rounds in simple graphs.

\begin{proposition}\label{cor:lower_bound_simple}
In simple graphs with arboricity 2, computing a $2$-forest-decomposition with success probability at least $0.5$ requires $\Omega(n)$ rounds.
\end{proposition}
\begin{proof}
First construct $G$ with parameters $\arb=2$ and $t = n/10$  (See Figure \ref{fig:lower_bound_simple_1}). Next, replace every set of parallel edges by a copy of the complete graph $K_4$ (See Figure \ref{fig:lower_bound_simple_2}). The resulting simple graph $H$ has $6 t+8 \leq n$ nodes and has arboricity 2.  

Suppose now that $A$ is an oblivious randomized algorithm which runs in $n/100$ rounds and assigns colors $\{1, 2 \}$ to the edges of $H$. Let $q$ denote the probability that algorithm $A$ outputs color 1 on edge  $(x_1,x_2)$. Since the two edges are $t$-hops away and $n/100<t/2$, the probability that $A$ outputs color 1 on $(y_1,y_2)$ is also $q$, independent of the color of $(x_1,x_2)$. Therefore, with probability at least $q (1-q) + (1-q) q \geq 1/2$, the edges $(x_1, x_2)$ and $(y_1, y_2)$ receive the same color; in this case, that color does not form a forest.
\end{proof}

\begin{figure}[H]
\centering
\begin{subfigure}[b]{0.49\textwidth}
\centering
\includegraphics[scale=0.65]{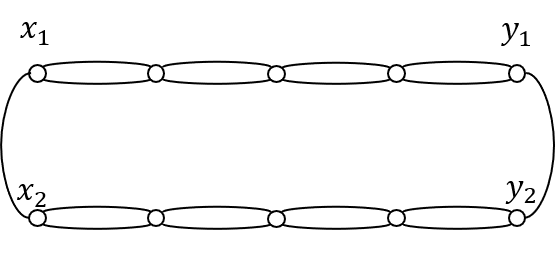}
\caption{}\label{fig:lower_bound_simple_1}
\end{subfigure}
\begin{subfigure}[b]{0.49\textwidth}
\centering
\includegraphics[scale=0.65]{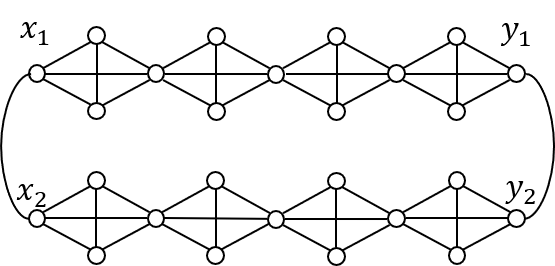}
\caption{}\label{fig:lower_bound_simple_2}
\end{subfigure}
\caption{}\label{fig:lower_bound_simple}
\end{figure}

\section{Acknowledgements}
Thanks to Vladimir Kolmogorov for suggesting how to set the parameters for Lemma~\ref{sclem1}. Thanks to Noga Alon for explaining some lower bounds for star arboricity. Thanks to Louis Esperet for some suggestions on notations and terminology. Thanks to anonymous conference and journal reviewers for helpful comments and suggestions.

\appendix

\section{Proof of Theorem~\ref{ace1prop}}
\label{ace1propapp}

For the first result, starting with $i=1$, we remove vertices with degree at most $t = \lfloor (2+\eps)\arb^{*} \rfloor$ and add them to $H_i$. We continue this process, forming sets $H_1, \dots, H_k$, until the graph is empty.  

We claim that each iteration removes at least a $(\eps / (2+\eps) )$ fraction of the vertices in the remaining graph. For, consider the remaining graph $G' = (V', E')$. If  more than $\frac{2}{2+\eps} |V'|$ vertices have degree greater than $(2+\eps)\arb^{*}$, then we derive a contradiction as follows:
\begin{align*}
2|E'| = \sum_{v \in V'} \deg(v) > ((2+\eps)\arb^{*})|V'|\cdot\frac{2}{2+\eps} = 2\arb^{*} |V'| \geq 2 ({|E'|}/{|V'|}) \cdot |V'| = 2|E'|.
\end{align*}

Since the number of vertices reduces by a $(1-\Omega(\eps))$ factor in each iteration, there are at most $k = O( \log n / \eps)$ iterations in total. 

For the second result, consider an edge $e = uv$ where $u \in H_i$ and $v \in H_j$. If $i < j$, we orient it from $u$ to $v$ and likewise if $i > j$ we orient it from $v$ to $u$.   If $i=j$, we orient it from the vertex with a lower ID to the vertex  with a higher ID.  Since a vertex in $H_i$ has at most $t$ neighbors in $H_i \cup \ldots \cup H_k$, the outdegree is at most $t$.  Since  the edges are always oriented from a lower index partition to a higher index, with ties broken by vertex ID, the resulting orientation is acyclic.

For the third result, we arbitrarily give distinct labels to the out-edges of each vertex; this gives us a $t$-forest decomposition where, moreover, each tree in each forest is rooted. We can use the algorithm of \cite{cole1986deterministic} to get a proper 3-vertex-coloring of each tree in $O(\log^* n)$ rounds. If we assign each edge to the color of its parent, then each of the $t$ forests decomposes into $3$ star-forests.

For the final result, consider the following process: we first fix some acyclic $t$-orientation. For each vertex $v$, we go through its out-edges $e_1, \dots, e_t$ in some arbitrary order, and each edge $e_i$ in turn chooses a color from its palette not already chosen by edges $e_1, \dots, e_{i-1}$.  This can be done in a greedy fashion since every edge has a palette of size  $t$. Each vertex operates independently, so the entire process can be simulated in $O(1)$ rounds. 

\section{Proof of Theorem~\ref{lem:Hpartition3}}
\label{Hpart3-sec}
To start, apply Theorem~\ref{ace1prop}(1) with $\eps/10$ in place of $\eps$, giving partition $H_1, \dots, H_k$ for $k = O( \frac{\log n}{\eps} )$ where each vertex $v \in H_i$ has at most $t = \lfloor (2 + \eps/10) \arb^* \rfloor$ neighbors in $H_i \cup \dots \cup H_k$. We orient the edges from $H_i$ to $H_j$ if $i < j$ and otherwise break tie by vertex ID. 

We proceed  for iterations $j = k, k-1, \dots, 1$; at iteration $j$, we define $E_j$ to be the set of edges which have one endpoint in $H_j$ and the other endpoint in  $H_j \cup \dots \cup H_k$. We also define a related line graph $G'_j$ as follows: each edge $e$ in $E_j$ corresponds to a node $x_e$ in $G'_j$. For every pair of edges $e, e'$ in $E_j$ which share a common vertex in $H_j$, there is an edge in $G'_j$ between corresponding nodes $x_e$ and $x_{e'}$.  Our strategy is to compute a proper list vertex-coloring of each $G'_j$, where the residual palette $Q'(x_e)$ for edge $e = uv \in E_j$ is obtained by removing from $Q(e)$ any colors chosen already by any out-edges of $u$ or $v$ in $E_{j+1} \cup \dots \cup E_k$.

We first claim that this procedure gives an LSFD, where each edge is oriented toward the center of the star. For, suppose  edges $e = uv, e' = u'v$ share a color $c$ and intersect in vertex $v$, but $e$ is oriented away $v$. Say $v \in H_i, u \in H_j, u' \in H_{j'}$; necessarily $j\geq i$ and $e \in H_i$ by definitions. If $j' \geq i$, then $e' \in H_{i}$ and the graph $G'_j$ would have an edge between $x_e$ and $x_{e'}$, and they could not receive the same color. If $j' < i$, then $e' \in H_{j'}$ and the color chosen by $e$ would have been removed from $Q'(x_{e'})$ in iteration $j'$.  Either case is a contradiction.

Next let us argue that each graph $G'_j$ can be greedily colored, i.e. for each edge $e = uv$, the palette of each node $x_e$ is larger than its degree. Suppose that $u$ and $v$ have $a$ and $b$ many neighbors in $H_{j+1} \cup \dots \cup H_k$ respectively. Then $x_e$ has at most $2 t - 1 - (a + b)$ neighbors in $G'_j$ and has $|Q'(x_e)| \geq |Q(e)| - (a+b)$. So $G'_j$ has maximum degree $\Delta(G'_j) = 2 t  = \Theta(\arb^*)$ and in either case the node $x_e \in G'_j$ satisfies
$$
| Q'(x_e) | - \text{deg}(x_e) \geq |Q(e)| - (2t - 1) \geq \lfloor (4 + \eps) \arb^*  \rfloor - 2 \lfloor (2 + 0.1 \eps) \arb^*  \rfloor + 1 \geq 1 + \lfloor \eps \arb^* / 2 \rfloor~.
$$

Finally, let us analyze the round complexity. We use an algorithm of \cite{EPS15} for list-vertex coloring with palettes of size $\deg + \eps \Delta$  applied to each graph $G'_j$.  There are three parameter regimes to consider. First, if $\arb \leq n^3$, then the algorithm of \cite[Corollary 4.1]{EPS15} combined with the network decomposition result of \cite{RG20}, runs in $O( \log(1/\eps) ) + \polyloglog m$ rounds; since $m \leq n \arb \leq n^4$, this is  $O( \log(1/\eps) )+ \polyloglog n$.

Next, when $\arb > n^3$, we use the algorithm of \cite[Theorem 4.1]{EPS15}. This algorithm requires $\eps \Delta(G'_j) > (\log  |V(G'_j)|)^{1 + \Omega(1)}$. Here, $|V(G'_j)| = m$ and $\Delta(G'_j) = \Theta(\arb) $, and $\eps > 1/n$ by assumption, so $\frac{\eps \Delta(G'_j)}{\log^2 |V(G'_j)|} \geq \Omega(\frac{\arb}{n \eps \log^2 m}) > \omega(1)$. So the algorithm runs in $O(\log^* m + \log(1/\eps))$ rounds.

Finally, when $\arb$ is super-exponentially larger than $n$, we can obtain an $(O(\log n), O(\log n))$-network decomposition of $G^3$. We then color each $E_j$ by iterating sequentially through the classes; within each cluster, we simulate a greedy coloring of the edges. Since clusters are  at distance at least $3$, no edges will choose a conflicting color. The overall runtime is $O(\log^2 n)$.

Putting the three cases together, we can color each $G'_j$ with round complexity $$
O \bigl( \min \bigl\{ \log(1/\eps) + \log^* m,  \log(1/\eps) + \polyloglog n, \log^2 n \bigr\} \bigr)
$$
This process is repeated for iterations $j = 1, \dots, k = O( \frac{\log n}{\eps})$.

\section{Proof of Proposition~\ref{decompose-prop3}}
\label{dp3-app}

We begin with the following algorithm to reduce diameter in a given list-forest-decomposition $\phi$.
 \begin{algorithm}[H]
\caption{\textsc{Reduce\_Forest\_Diameter}$(\eps, \phi)$}\label{alg:reduce_diameter}
\begin{algorithmic}[1]\small
\State Initialize $E_0 \leftarrow E$ and $E_1', E_1'' \leftarrow \emptyset$.
\State Apply Theorem~\ref{ace1prop}(2) to get an  acyclic $3 \arb$-orientation $J$ of $G$.
\For{\textbf{each} vertex $v$}
\State Draw independent Bernoulli-$1/2$ random variable $X_v$
\If{$X_v = 1$}
\For{$j = 1$ to $k'$ where $k' = \lceil \eps \arb / 20 \rceil$}
\State Select an edge $e_{v,j}$ uniformly at random from the out-neighbors of $v$ in $J$.
\State  Remove edge $e_{v,j}$ from $E_0$, add it to $E_1'$.
\State Set $\phi_1'(e_{v,j}) = j$.
\EndFor
\EndIf
\EndFor
\For{\textbf{each} color $i \in \cs$ and each vertex $v$}
\If{there is a $i$-colored directed path of length $10000 \log n / \eps$ starting from $v$ in $E_0$}
\State{Remove  all edges of color $i$ incident to $v$ from $E_0$, add them to $E_1''$.}
\EndIf
\EndFor
\For{\textbf{each} $j = 1, \dots, k'$  and each vertex $v$}
\If{there is a $j$-colored directed path of length $10000 \log n / \eps$ starting from $v$ in $E_1'$}
\State{Remove  all edges of color $j$ incident to $v$ from $E_1'$, add them to $E_1''$.}
\EndIf
\EndFor
\State Apply Theorem~\ref{ace1prop}(3) to obtain a $3 \lfloor 2.01 \arb^*(E_1'') \rfloor$-star-forest-decomposition $\phi_1''$ of $E_1''$.
\State Return $E_0$ and $E_1 = E_1' \cup E_1''$ along with forest decomposition $\phi_1 = \phi_1' \cup \phi_1''$ of $E_1$.
\end{algorithmic}
\end{algorithm}

\begin{proposition} 
\label{decompose-prop3a1}
Given a multigraph $G$ with a LFD $\phi$ and $\eps > 0$, Algorithm~\ref{alg:reduce_diameter} runs in $O( \frac{\log n}{\eps})$ rounds. It partitions the edges as $E = E_0 \sqcup E_1$, where $\phi_1$ is an $\lceil \eps \alpha \rceil$-FD of $E_1$. W.h.p, both the restriction of $\phi$ to $E_0$ and the decomposition $\phi_1$ on $E_1$ have diameter $D \leq O( \frac{\log n}{\eps} )$.
\end{proposition}
\begin{proof}
The complexity follows from specifications of Theorem~\ref{ace1prop}.   Since the orientation is acyclic, $\phi_1'$ is a $k'$-forest decomposition of $E_1'$ for $k' = \lceil \eps \arb/20 \rceil$.  The two modification steps (Lines 10 -- 12 and 13 -- 15) ensure each forest in $E_0$ or $E_1'$ has diameter $O( \frac{\log n}{ \eps})$, while the forests in $E_1''$ have diameter $2$.  Overall, $E_1$ is decomposed into $k' + 3 \lfloor 2.01 \arb^*(E_1'') \rfloor$ forests; we can upper-bound this, somewhat crudely, as $k' + 7 |E_1''|$.

It remains to bound $|E_1''|$. We first claim that any edge goes into $E_1''$ with probability at most $1/n^{24}$. For, suppose that $e$ remains in $E_0$ with color $i$. Every other color-$i$ edge  gets removed with probability at least $ \frac{1}{2} \times \frac{ \lceil \eps \arb / 20 \rceil }{3 \arb} \geq 0.005 \eps$, and edges at distance two are independent. Thus, any path of length $r \geq 10000 \log n / \eps$ path survives with probability at most $(1 - 0.005 \eps)^{r/2} \leq e^{-0.005 \cdot 10000 \cdot 1/2 \cdot \log n} = n^{-25}$, and there are at most $n$ paths of color $i$ involving any given edge.

Similarly, suppose that $e$ goes into $E_1'$ with color $j \in \{1, \dots, k' \}$. Each vertex only has deleted out-neighbors with probability $1/2$. Thus, starting at the edge $e$, and following its directed path with respect to the $j$-colored edges, there is a probability of $1/2$ of stopping at each vertex. Thus, the probability that $e$ goes goes into $E_1''$ is at most $ 2^{-10000 \log n / \eps} \leq n^{-24}$. 

Putting these two cases together, $E_1''$ has expected size at most $m / n^{24} \leq \arb / n^{23}$. By our assumption that $\eps > 1/n$, Markov's inequality gives $\Pr(  |E''_1| \geq \eps \arb / 20 ) \leq O( \frac{ \arb/n^{23}}{\eps \arb} ) \leq O(n^{-22})$.  So w.h.p $E_1$ uses at most $\lceil \eps \arb / 20 \rceil + 7 ( \eps \arb / 20  ) \leq \lceil \eps \arb \rceil$ forests.
\end{proof}

We next show that the diameter can be reduced further in some cases. Note that Proposition~\ref{decompose-prop3a2}, in its full generality for list-forest-decompositions, will not be directly required for our algorithm.
\begin{proposition} 
\label{decompose-prop3a2}
Let $G$ be a multigraph with an LFD $\phi$. If $|\cs| = \rho \arb$ and $\arb \geq \Omega \bigl( \min \{ \frac{\log n}{\eps}, \frac{\rho \log \Delta}{\eps^2} \} \bigr)$, there is an $O( \frac{ \rho \log n}{\eps})$-round algorithm to obtain an edge-set $E'$ such that $\arb^*(E') \leq \eps \arb$ and such that w.h.p. $\phi$ has diameter $D \leq O( \rho / \eps )$ in the graph $G \setminus E'$.
\end{proposition}
\begin{proof}
First, by applying Proposition~\ref{decompose-prop3a1} with $\eps/4$ in place of $\eps$, we reduce the diameter of the forests to $O( \frac{\log n}{\eps})$, while discarding an edge set of pseudo-arboricity at most  $k' \leq \lceil \eps \arb / 4 \rceil$. In $O( \frac{\log n}{\eps} )$ rounds we can choose an arbitrary rooting of each remaining tree. For each color $c$, let $U_c$ be the set of vertices whose depth is a multiple of $N = \lceil 4 \rho/ \eps \rceil$.

Consider the following random process:   For each color $c$ and each vertex $u \in U_c$, we independently draw an integer $J_{u,c}$ uniformly at random from $[N]$. For all vertices $v$ of depth $J_{u,c}$ below $u$ in the color-$c$ tree, we remove the color-$c$ parent edge of $u$.  After this deletion step,  every vertex $u \in U_c$ is disconnected from its distance-$N$ descendants and is also disconnected from its distance-$N$ ancestors. Thus,  with probability one, the forests have diameter reduced to $4N \leq O(\rho/\eps)$.

It remains to bound the pseudo-arboricity of the deleted edges. For each vertex $v$, let $B_v$ denote the bad-event that $v$ has more than $k'' = \eps \arb / 2$ deleted parent edges. If all these bad-events are avoided, then the deleted edges have pseudo-arboricity at most $k' + k'' \leq \eps \arb$ from both stages.

For a vertex $v$ and color $c$, let $u_c$ be the maximum-depth ancestor of $v$ in $U_c$. The color-$c$ parent of $v$ gets deleted if and only if $J_{u_c,c}$ is equal to the depth of $v$ below $u_c$. This has probability  $1/N$, so the expected number of deleted parents is at most $|\cs|/N \leq \eps \arb/4$. Each color operates independently, so by Chernoff's bound we have $\Pr(B_v) \leq F_+(\eps \arb/4, \eps \arb/2) \leq e^{-\eps \arb/12}$. 

If $\eps \arb \geq \Omega(\log n)$, then w.h.p. none of the bad-events occur, and we are done.  If $\eps^2 \arb \geq \Omega(\rho \log \Delta)$ we use the LLL algorithm of \cite{chung2017distributed}.  We have already calculated $p \leq e^{-\eps \arb / 12}$. Also, events $B_v$ and $B_w$ only affect each other if $v$ and $w$ have distance at most $2 N$, hence $d \leq \Delta^{2 N} \leq \Delta^{4 \rho/\eps + 2}$.   So $p d^2 \leq e^{-\eps \arb/12}  (\Delta^{4 \rho/\eps + 2})^2 \ll 1$ for $\eps^2 \arb \geq \Omega(\rho \log \Delta)$. Each $B_v$ depends on vertices within distance $O(\rho/\eps)$, so the LLL algorithm can be simulated in $O( \frac{\rho \log n}{\eps} )$ rounds on $G$.
\end{proof}

To show Proposition~\ref{decompose-prop3}, suppose now we are given some $k$-FD of $G$; we may assume that $k \leq O(a)$, as we can always use Theorem~\ref{ace1prop} to obtain a $2.01$-FD. Here,  we have $|\cs| = k$ and so $\rho = k/a \leq O(1)$.   For  the bound $D \leq O( \frac{\log n}{\eps})$, we can directly apply Proposition~\ref{decompose-prop3a1}. For the bound $D \leq O( \frac{1}{\eps} )$ when $a$ is large, we apply Proposition~\ref{decompose-prop3a2} with $\eps/10$ in place of $\eps$ to obtain a $k + \lceil \eps \arb / 10 \rceil$-FD, where the uncolored edges $E'$ have $\arb^*(E') \leq \eps \arb / 10$. We then use Theorem~\ref{ace1prop} to obtain a $3 \arb^*(E')$-SFD of $E'$. Overall, this gives a $k + \lceil \eps \arb \rceil$ FD of $G$ of diameter $O( 1/ \eps)$.

\section{Concentration bound for Lemma~\ref{sclem2}}
Here, we show the concentration bound we used in the proof of Lemma~\ref{sclem2}.
\begin{proposition}
\label{nnb4}
Let $s = a(1+100 \eps), t = a(1+ \eps)$. Let $X$ be any fixed subset of $J(v)$ and let $x = |X| \leq t$. Suppose the hypotheses of Lemma~\ref{sclem2} hold, and that $C_v$ is fixed so that $|Q(uv) \cap C_v| \geq s$ for all $u \in X$.  Then the probability that $X$ has fewer than $x$ neighbors in $W_v$ is at most
$$
\Bigl(  \frac{s e^{-\eps x}}{s - x } \Bigr)^{s - x } \Bigl( \frac{s (1 - e^{-\eps x})}{x} \Bigr)^x
$$
\end{proposition}
\begin{proof}
For each $c \in C_v$, let $z_c$ be the number of vertices $u \in X$ with $c \in Q(uv)$, and let $Y_c$ be the indicator function that $c \in \bigcup_{u \in X} (Q(uv) \setminus C_u)$.  Here $\Pr(Y_c = 1) = 1 - (1 - \eps)^{z_c} \geq 1 - e^{-\eps z_c}$. By hypothesis  we have $\sum_{c \in C_v} z_c = \sum_{u \in X} |Q(uv) \cap C_v| \geq x s$. 

Consider variable $Y = \sum Y_c$, and note that $X$ has fewer than $x$ neighbors if and only if $Y < x$. For some parameter $\theta \in [0,1]$ to be determined, we define the random variable
$$
\Phi = (1-\theta)^Y = \prod_{c \in C_v} (1 - \theta)^{Y_c}.
$$
Since the colors are independent, we have
\begin{equation}
\label{phieq1}
\mathbf E[ \Phi ] = \prod_{c \in C_v} (1 - \theta \Pr(Y_c = 1)) \leq \prod_{c \in C_v} (1 - \theta (1-e^{-\eps z_c})) = e^{\sum_{c \in C_v} \log 
 (1 - \theta (1-e^{-\eps z_c})) }
\end{equation}

Elementary calculus shows that the function $y \mapsto \log \bigl( 1 - \theta (1-e^{-y}) \bigr)$ is  negative, decreasing, and concave-up. Since  $0 \leq z_c \leq x$, we thus bound it by the secant line from $0$ to $x$, i.e.
$$
\log (1 - \theta (1-e^{-\eps z_c}))  \leq \frac{z_c}{x} \log(1 - \theta (1 - e^{-\eps x}))~.
$$

Substituting this bound into Eq.~(\ref{phieq1}) and using the bound $\sum z_c \geq x s$, we calculate
$$
\mathbf E[ \Phi ] \leq e^{\sum_{c \in C_v} \frac{z_c}{x} \log (1 - \theta (1-e^{-\eps x})) } = \bigl( 1 - \theta(1-e^{-\eps x}) \bigr)^{\sum_{c \in C_v} z_c / x} \leq  \bigl( 1 - \theta(1-e^{-\eps x}) \bigr)^{s}~.
$$

Now by Markov's inequality applied to $\Phi$, we get
\begin{equation}
\label{phieq0}
\Pr(Y < x) \leq \mathbf E[ \Phi ]  (1-\theta)^{-x} \leq \bigl( 1 - \theta(1-e^{-\eps x}) \bigr)^{s} \bigl( 1 - \theta \bigr)^{-x}~.
\end{equation}

At this point, we set 
$$
\theta = \frac{ (1 - e^{-\eps x}) s - x }{ (1-e^{-\eps x}) (s - x)}~.
$$

Clearly $\theta \leq 1$. We claim also that $\theta \geq 0$, which substituting into Eq.~(\ref{phieq0}) will give the claimed formula.  For, consider the function $f(x) = (1 - e^{-\eps x})  s  -  x$. The second derivative is given by $f''(x) = -e^{-\eps x} \eps^2 s \leq 0$. Hence, the minimum value of  $f(x)$ in the region $x \in [0, t]$ will occur at either $0$ or $t$. For the former, we have $f(0) = 0$. For the latter, we have
$$
f(t) = (1 - e^{-\eps t}) s - t = (1 - e^{-\eps a (1 + \eps)}) a (1 + 100 \eps) - a (1 + \eps)
$$

Now $\eps a \geq 10^6 \log \Delta$, so $e^{-\eps a (1+\eps)} \leq \Delta^{-10^6}$. Since $a \leq \Delta$ and $\Delta \geq 2$, we get
\[
f(t) \geq  (1 - \Delta^{-10^6}) a (1 + 100 \eps) - a (1 + \eps) = a ( 99 \eps - 100 \eps / \Delta^{10^6} - a / \Delta^{10^6}) \geq 0
\]
as desired.
\end{proof}

\section{Miscellaneous Observations and Formulas}
\label{sec:missingproofs}

\begin{proposition}
\label{eps-bound}
For any integer $\arb \geq 1$ and any $\eps > 0$, there is a multigraph with arboricity $\arb$ and $n = O(1/\eps)$ and $\Delta = O(\arb)$, for which any $\arb(1+\eps)$-FD has diameter $\Omega(1/\eps)$.
\end{proposition}
\begin{proof}
Consider the graph $G$ with $\ell = \lceil 1/\eps \rceil$ vertices arranged in a path and $\arb$ edges between consecutive vertices. This has maximum degree $\Delta = 2 \arb$ and arboricity $\arb$.  In any forest decomposition of diameter $D$, each forest must consist of consecutive sub-paths each of length at most $D$. Thus, each color uses at most $\lceil \frac{\ell}{D+1} \rceil \times D \leq D (1 + \ell/D+1)$ edges. There are $(\ell-1) \arb$ total edges, so we must have $\arb (1+\eps) D (1 + \ell/(d+1)) \geq (\ell - 1) \arb$. Since $\ell = \Theta(1/\eps)$, this implies that $D \geq \Omega(1/\eps)$.
\end{proof}

\begin{proposition}
\label{star-pseudo-cycle}
For a loopless multigraph, there holds $\arb_{\text{star}} \leq 2 \arb^*$.
\end{proposition}
\begin{proof}
It suffices to show that a loopless pseudo-tree can be decomposed into two star-forests. This pseudo-tree can be represented as a cycle $C = \{x_1, \dots, x_t \}$ plus trees $T_1, \dots, T_t$ rooted at each $x_i$.

We can two-color the edges of $C$ from $\{0, 1 \}$, such that there is at most one pair of consecutive edges on $C$ with the same color; say w.l.g. it is color $0$. For each edge $e$ in a tree $T_i$ which is at depth $d$ ($d = 0$ means that $e$ has an endpoint $x_i$), we assign $e$ to color $(d+1) \mod 2$.  In particular, the child edges from root $x_i$ have color $1$, the grandchild edges have color $0$, etc.
\end{proof}

\begin{theorem}
\label{degen-thm}
For a multigraph with degeneracy $\delta$, there holds $\arb_{\text{star}}^{\text{list}} \leq 2 \delta \leq \min \{4 \arb^*, 4 \arb - 2 \}$.
\end{theorem}
\begin{proof}
It is a standard result that $\delta \leq 2 \min \{ \arb - 1, \arb^* \}$, so it suffices to show $\arb_{\text{star}}^{\text{list}} \leq 2 \delta$.

Fix some acyclic $\delta$-orientation of $G$, and color each edge sequentially, going backward in the orientation. For each edge $e$ oriented from $u$ to $v$, we choose a color in $Q(e)$ not already chosen by any out-neighbor of $u$ or $v$. Each vertex has at most $\delta$ out-neighbors, so at most $\delta$ colors are already used  by out-neighbors of $v$ and $\delta-1$ colors used by out-neighbors of $u$. Since $e$ has a palette of size $2\delta$, we can always choose a color for $e$. The resulting coloring at the end  is a star-list-coloring, where each edge is oriented toward the center of its star.
\end{proof}

\bibliographystyle{alpha}	
\bibliography{references}

\end{document}